\documentclass[english]{scrartcl}
\usepackage[T1]{fontenc}
\usepackage[utf8]{inputenc}
\usepackage{geometry}
\usepackage[active]{srcltx}
\usepackage{color}
\usepackage{babel}
\usepackage{enumitem}
\usepackage{bm}
\usepackage{amsmath}
\usepackage{amsthm}
\usepackage{amssymb}
\usepackage{float}
\usepackage{esint}
\usepackage{listings}
\usepackage{mathtools}
\usepackage[numbers]{natbib}
\usepackage[unicode=true,pdfusetitle,bookmarks=true,bookmarksnumbered=true,bookmarksopen=true,bookmarksopenlevel=1,
 breaklinks=true,pdfborder={0 0 1},backref=false,colorlinks=true]
 {hyperref}
\hypersetup{
 linkcolor=blue, citecolor=blue, urlcolor=blue, filecolor=blue,pdfpagelayout=OneColumn, pdfnewwindow=true,pdfstartview=XYZ, plainpages=false}

\makeatletter
\theoremstyle{plain}
\newtheorem{thm}{\protect\theoremname}
\theoremstyle{plain}
\newtheorem{lem}{\protect\lemmaname}
\theoremstyle{remark}
\newtheorem{rem}{\protect\remarkname}
\theoremstyle{plain}
\newtheorem{assumption}{\protect\assumptionname}
\theoremstyle{plain}
\newtheorem{prop}[thm]{\protect\propositionname}
\theoremstyle{plain}
\newtheorem{cor}{\protect\corollaryname}
\theoremstyle{remark}
\newtheorem{notation}{\protect\notationname}

\usepackage{scrlayer}
\usepackage{lastpage}
\usepackage{amsmath}
\usepackage{amssymb}
\usepackage{graphicx}
\usepackage{fancybox} 
\usepackage{moreverb} 
\usepackage{listings} 
\usepackage{backref}
\usepackage{bm}
\usepackage{array}
\usepackage{makecell}
\usepackage{multirow}
\usepackage{authblk}
\usepackage{fontawesome5}
\usepackage{lipsum}
\usepackage{subcaption}
\usepackage{stackrel}
\usepackage{scalefnt}

\usepackage{ifpdf} 
\ifpdf 

 \IfFileExists{lmodern.sty}{\usepackage{lmodern}}{}

\fi 

\let\myTOC\tableofcontents
\renewcommand\tableofcontents{%
  \pdfbookmark[1]{\contentsname}{}
  \myTOC
}

\def\LyX{\texorpdfstring{%
  L\kern-.1667em\lower.25em\hbox{Y}\kern-.125emX\@}
  {LyX}}

\@addtoreset{footnote}{section}

\renewcommand*{\backref}[1]{}
\renewcommand*{\backrefalt}[4]{%
   \ifcase #1 
    \or 
      (Cited on page~#2)%
   \else
      (Cited on pages~#2)
    \fi} 

\usepackage{dsfont}
\usepackage{bbm}

\usepackage[dvipsnames]{xcolor}
\definecolor{blue}{HTML}{1F77B4}
\definecolor{orange}{HTML}{FF7F0E}
\definecolor{green}{HTML}{2CA02C}
\definecolor{red}{HTML}{D62728}
\definecolor{purple}{HTML}{9467BD}
\definecolor{brown}{HTML}{8C564B}
\definecolor{pink}{HTML}{E377C2}
\definecolor{grey}{HTML}{7F7F7F}
\definecolor{yellow}{HTML}{BCBD22}
\definecolor{cyan}{HTML}{17BECF}
\definecolor{turquoise}{HTML}{3FE0D0}

\usepackage[algo2e,ruled,vlined]{algorithm2e}
\SetKwInput{KwIn}{Inputs}\SetKwInput{KwOut}{Outputs}
\usepackage{xcolor}
\definecolor{algoColorKeyword}{named}{blue}
\definecolor{algoColorComment}{named}{olive}

\SetKwComment{Comment}{$\triangleright$\ }{}

\usepackage{enumitem}
\setlist{leftmargin=*, topsep=0.5em, parsep=0pt, itemsep=1em, labelindent=0pt, align=left}

\makeatother

\providecommand{\assumptionname}{Assumption}
\providecommand{\corollaryname}{Corollary}
\providecommand{\lemmaname}{Lemma}
\providecommand{\notationname}{Notation}
\providecommand{\propositionname}{Proposition}
\providecommand{\remarkname}{Remark}
\providecommand{\theoremname}{Theorem}

\usepackage{orcidlink}
\definecolor{orcidlogocol}{HTML}{2E7E9F}

\begin{document}
\title{Signed random Fourier features for fast density estimation with indefinite kernels}

\author[1]{\orcidlinki{\textcolor{black}{Xie Wang}}{0009-0006-6007-2702}} 
\author[2]{\orcidlinki{\textcolor{black}{Nicolas Langren\'e}}{0000-0001-7601-4618}\thanks{Corresponding author, nicolaslangrene@bnbu.edu.cn}} 
\author[2]{\orcidlinki{\textcolor{black}{Wen Chen}}{0000-0002-7268-0979}}
\affil[1]{\normalsize Mathematical Institute, University of Oxford, Oxford, OX2 6GG, UK}
\affil[2]{\normalsize Guangdong Provincial/Zhuhai Key Laboratory of Interdisciplinary Research and Application for Data Science, Beijing Normal-Hong Kong Baptist University, Zhuhai 519087, China.}

\date{\today}

\maketitle
\begin{abstract}
Kernel density estimation (KDE) is one of the most fundamental statistical
estimators of density functions. Its direct implementation on a dataset
of $N$ points incurs an $\mathcal{O}(N^{2})$ computational cost, which is prohibitive for large-scale datasets. Kernel approximation
techniques can be applied to bring the computational cost down to
$\mathcal{O}(N)$. The random Fourier features (RFF) technique, based
on sampling from the spectral density of the kernel function, has
become popular to speed up kernel estimators for machine learning
applications. Unfortunately, it is restricted to positive definite
kernels, while the majority of kernel functions popular in KDE, such
as the parabolic kernel, do not satisfy this property. To overcome this limitation, this article introduces the signed random
Fourier features (SRFF) technique. It is a generalization of RFF compatible
with indefinite kernels whose inverse Fourier transform is absolutely
integrable. The motivation for introducing this method is to speed
up KDE in the case of multivariate compact kernels, which are generally not positive definite. We detail how to implement SRFF for both product kernels and isotropic kernels. For the class of Kuttner-Golubov kernels $K(\boldsymbol{x}_{i},\boldsymbol{x}_{j})=(1-\left\Vert \boldsymbol{x}_{i}-\boldsymbol{x}_{j}\right\Vert ^{\alpha})^{\beta}\mathbbm{1}_{\{\left\Vert \boldsymbol{x}_{i}-\boldsymbol{x}_{j}\right\Vert \leq1\}}$
where $\boldsymbol{x}_{i}\in\mathbb{R}^{d}$, $\boldsymbol{x}_{j}\in\mathbb{R}^{d}$,
$\alpha>0$, $\beta>0$, which includes the triangular, parabolic,
biweight, triweight, and other kernel functions of interest for KDE
as particular examples, we provide an explicit acceptance-rejection algorithm
to sample from its signed spectral density. Our numerical tests on a dataset of one million points confirm the computational efficiency and accuracy of SRFF for large-scale KDE.\\

\end{abstract}

\section{Introduction}

Kernel-based estimators, such as kernel density estimators, are a
cornerstone of statistical science. Their implementation, however,
suffers from scalability issues: the computation of all the pairwise
kernel values $K(\boldsymbol{x}_{i},\boldsymbol{x}_{j})$ for a kernel
function $K$ on a dataset $\boldsymbol{x}_{1},\ldots,\boldsymbol{x}_{N}$
of $N$ points valued in $\mathbb{R}^{d}$, $1\leq i,j\leq N$, requires
$\mathcal{O}(N^{2})$ operations and $\mathcal{O}(N^{2})$ memory
space, which can quickly become prohibitive for large-scale datasets.
Popular approximations of kernel functions in the field of machine
learning, such as the Nystr\"om method \citep{nystrom1930praktische,williams2000nystrom}
or random Fourier features \citep{rahimi2007random} can be used to
reduce these computational and memory costs. However, these approximation
techniques require the kernel function $K$ to be positive definite
(Subsection~\ref{subsec:rff}). In the context of kernel density
estimation, this is a strong limitation. With a few exceptions, such
as the Gaussian kernel, most kernel functions of interest in KDE,
such as the parabolic kernel \citep{epanechnikov1969nonparametric},
are not positive definite, making these modern kernel decomposition
techniques inapplicable.

When a kernel function $K$ is positive definite, the inverse Fourier transform
of $K$ is proportional to a density $f$, called spectral density
(Bochner's theorem, Subsection~\ref{subsec:rff}). The RFF technique
is constructed upon simulations from this spectral density. When $K$
is not positive definite, such as the parabolic kernel $K(\boldsymbol{x}_{i},\boldsymbol{x}_{j})=(1-\left\Vert \boldsymbol{x}_{i}-\boldsymbol{x}_{j}\right\Vert ^{2})\mathbbm{1}_{\{\left\Vert \boldsymbol{x}_{i}-\boldsymbol{x}_{j}\right\Vert \leq1\}}$,
the inverse Fourier transform $f$ still exists but is not a density anymore,
as it takes negative values (Subsection~\ref{subsec:symmetric_beta_kernels}).
From here, some articles such as \citep{pennington2015spherical}
propose to approximate $f$ by a tractable nonnegative function, effectively
overlooking the negative part of the spectral signed measure. Others,
such as \citep{liu2021fast,luo2021towards,he2024random},
propose to simulate $f_{+}/\left\Vert f_{+}\right\Vert _{1}$ and
$f_{-}/\left\Vert f_{-}\right\Vert _{1}$ separately, where $f_{+}=\max(f,0)$,
$f_{-}=\max(-f,0)$, $\left\Vert f\right\Vert _{1}=\int_{\mathbb{R}^{d}}|f(\boldsymbol{x})|d\boldsymbol{x}$,
provided $f$ is absolutely integrable (and therefore associated with
a finite signed measure), and to apply the weights
$\left\Vert f_{+}\right\Vert_{1} $ and $\left\Vert f_{-}\right\Vert_{1} $
to these simulated frequencies in the RFF estimator. One downside
is that these weights are not analytical and their estimation may
require costly computations.

In this article, we follow the latter approach, and improve it along
several directions. Our key contribution is the introduction of the
\textit{signed random Fourier features} (SRFF) technique (Subsection~\ref{subsec:srff}),
constructed upon simulations from the absolute spectral density $\left|f\right|/\left\Vert f\right\Vert _{1}$,
using a carefully designed acceptance-rejection algorithm (Subsection~\ref{subsec:rejection_sampling}),
and attaching to each simulated frequency the sign of the spectrum at
that point. Compared to existing random features for indefinite kernels,
its implementation does not require to compute the weights $\left\Vert f_{+}\right\Vert _{1}$, $\left\Vert f_{-}\right\Vert _{1}$, or $\left\Vert f\right\Vert _{1}$,
and it is shown to be more efficient than sampling from $f_{+}/\left\Vert f_{+}\right\Vert _{1}$
and $f_{-}/\left\Vert f_{-}\right\Vert _{1}$ separately (Subsection~\ref{subsec:rejection_sampling}).
Then, we detail how to implement SRFF for fast kernel density estimation
using the large class of multivariate compact kernels $K(\boldsymbol{x}_{i},\boldsymbol{x}_{j})=(1-\left\Vert \boldsymbol{x}_{i}-\boldsymbol{x}_{j}\right\Vert ^{\alpha})^{\beta}\mathbbm{1}_{\{\left\Vert \boldsymbol{x}_{i}-\boldsymbol{x}_{j}\right\Vert \leq1\}}$,
$\alpha>0$, $\beta>0$, which contains many multivariate kernel functions
of interest in KDE, such as the triangular, parabolic, biweight, triweight,
and tricube kernels (\citep{silverman1986density,loader1999local,scott2015multivariate,chacon2018multivariate},
Subsection~\ref{subsec:choice_of_kernel}), none of which are positive
definite on multivariate datasets. While the existing theory of reproducing
kernel Kre\u{\i}n spaces (RKKS, \citep{ong2004learning,huang2017indefinite})
has shown that most kernel-based estimators still work in the case
of indefinite kernels, our work is, to our knowledge, the first application
of random features to the problem of multivariate kernel density estimation
with indefinite kernels. Finally, our numerical tests on a dataset
of $N=1{,}000{,}000$ points confirm the scalability offered by SRFF
for large-scale KDE tasks.

The paper is organized as follows. Section~\ref{sec:signed_monte_carlo}
introduces a generalization of the Monte Carlo simulation method to
integrals involving signed densities. Section~\ref{sec:signed_rff}
applies this extension to the inverse Fourier transform of kernel functions,
yielding an extension of random Fourier features (RFF) to indefinite
kernels, called signed random Fourier features (SRFF). Section~\ref{sec:fast_kde}
details how to use SRFF to speed the computation of kernel density
estimators. Section~\ref{sec:numerical} illustrates numerically
the speed and accuracy of the proposed method, and Section~\ref{sec:conclusion}
concludes the paper.

\section{Signed Monte Carlo\label{sec:signed_monte_carlo}}

In this section, we introduce a simple extension of the Monte Carlo integration method, called signed Monte Carlo, which is concerned with the sampling of finite signed measures. This technique will then be applied to kernel function approximation in the subsequent section~\ref{sec:signed_rff}.

\subsection{Signed Monte Carlo formula}

Let $\boldsymbol{X}=(X_{1},\ldots,X_{d})\in\mathbb{R}^{d}$ be a real-valued, $d$ -dimensional continuous random
vector with probability density function $f\colon \mathbb{R}^{d}\to\mathbb{R}$,
and let $g\colon\mathbb{R}^{d}\to\mathbb{R}$ be a
measurable function such that $fg$ 
is absolutely integrable. By elementary probability theory, we have
the following
\begin{equation}
\underset{\hspace{0.25em}\mathbb{R}^{d}}{\int}\hspace{-0.125em}f(\boldsymbol{x})g(\boldsymbol{x})d\boldsymbol{x}=\mathbb{E}\left[g(\boldsymbol{X})\right].\label{eq:EgX}
\end{equation}
This probabilistic representation~\eqref{eq:EgX} suggests the following
unbiased Monte Carlo estimator of the integral $\int_{\mathbb{R}^{d}}f(\boldsymbol{x})g(\boldsymbol{x})d\boldsymbol{x}$:
\begin{equation}
\frac{1}{M}\sum_{m=1}^{M}g(\boldsymbol{X}_{m})\label{eq:EgX_MC}
\end{equation}
where $\boldsymbol{X}_{1}$, $\ldots$, $\boldsymbol{X}_{M}$ are $M$ i.i.d. simulations from the distribution of $\boldsymbol{X}$. This is known as the Monte Carlo integration method. We now extend this construction to the case where the function $f$ is no longer a probability density function. In particular, $f$ is not nonnegative everywhere. From here onwards, we relax the assumption that $f$ is a probability density and only assume that $f$ is absolutely integrable over $\mathbb{R}^{d}$, with no non-negativity imposed. 

\begin{lem}\label{lem:f_decomposition}
Let $f\colon\mathbb{R}^{d}\to\mathbb{R}$ be an absolutely integrable function, not identically equal to zero. Then
\begin{equation}
f=\mathrm{sign}(f)\left\Vert f\right\Vert _{1}\overset{_{\,\frown}}{f}\label{eq:f_decomposition}
\end{equation}
where $\left\Vert f\right\Vert _{1}:={\int}_{\mathbb{R}^{d}}|f(\boldsymbol{x})|d\boldsymbol{x}$ is such that $0<\left\Vert f\right\Vert _{1}<\infty$, $\overset{_{\,\frown}}{f}:=\frac{|f|}{\left\Vert f\right\Vert _{1}}$ is a probability density function, and where the $\mathrm{sign}$ function is defined as
$$\mathrm{sign}(x):=
\begin{cases}
\hspace{0.75em}1 & ,\ x>0\\
\hspace{0.75em}0 & ,\ x=0\\
-1 & ,\ x<0
\end{cases}$$
\end{lem}
\begin{proof}
Simply use the fact that $f(\boldsymbol{x})=\mathrm{sign}(f(\boldsymbol{x}))|f(\boldsymbol{x})|$ for all $\boldsymbol{x}\in\mathbb{R}^d$, and then multiply and divide by $\left\Vert f\right\Vert _{1}$ on the right hand side, which, from the assumptions, exists, is positive and finite.
\end{proof}

\begin{lem}
\label{lem:signed_MC} Let $f\colon\mathbb{R}^{d}\to\mathbb{R}$
be an absolutely integrable function such that $\int_{\mathbb{R}^{d}}f(\boldsymbol{x})d\boldsymbol{x}\neq0$,
and let $g\colon\mathbb{R}^{d}\to\mathbb{R}$ be a
measurable function such that $fg$
is absolutely integrable. Then, the following holds
\begin{equation}
\underset{\hspace{0.25em}\mathbb{R}^{d}}{\int}\hspace{-0.125em}f(\boldsymbol{x})g(\boldsymbol{x})d\boldsymbol{x}=\bigg(\underset{\hspace{0.25em}\mathbb{R}^{d}}{\int}\hspace{-0.125em}f(\boldsymbol{x})d\boldsymbol{x}\bigg)\frac{\mathbb{E}\left[\mathrm{sign}(f(\boldsymbol{X}))g(\boldsymbol{X})\right]}{\mathbb{E}\left[\mathrm{sign}(f(\boldsymbol{X}))\right]}\label{eq:signed_MC}
\end{equation}
where $\boldsymbol{X}=(X_{1},\ldots,X_{d})\in\mathbb{R}^{d}$ is a continuous
random vector with probability density function $\overset{_{\,\frown}}{f}=|f|/\left\Vert f\right\Vert _{1}$, and where the $\mathrm{sign}$ function is defined in Lemma~\ref{lem:f_decomposition}.
\end{lem}

\begin{proof}
First, remark that the assumption $\int_{\mathbb{R}^{d}}f(\boldsymbol{x})d\boldsymbol{x}\neq0$
implies that $f$ is not uniformly equal to zero, and therefore
$\left\Vert f\right\Vert _{1}>0$. This 
ensures that $\overset{_{\,\frown}}{f}=|f|/\left\Vert f\right\Vert _{1}$ is a well-defined probability density function. Then, using Lemma~\ref{lem:f_decomposition},\\
\begin{equation}
\underset{\hspace{0.25em}\mathbb{R}^{d}}{\int}\hspace{-0.125em}f(\boldsymbol{x})g(\boldsymbol{x})d\boldsymbol{x}=\left\Vert f\right\Vert _{1}\!\underset{\hspace{0.25em}\mathbb{R}^{d}}{\int}\hspace{-0.125em}\mathrm{sign}(f(\boldsymbol{x}))g(\boldsymbol{x})\overset{_{\,\frown}}{f}(\boldsymbol{x})d\boldsymbol{x}=\left\Vert f\right\Vert _{1}\mathbb{E}\left[\mathrm{sign}(f(\boldsymbol{X}))g(\boldsymbol{X})\right]\label{eq:intfg}
\end{equation}
where the random vector $\boldsymbol{X}$ has density $\overset{_{\,\frown}}{f}$. Next, applying equation~\eqref{eq:intfg} with $g\equiv 1$ gives
\begin{equation}
\underset{\hspace{0.25em}\mathbb{R}^{d}}{\int}\hspace{-0.125em}f(\boldsymbol{x})d\boldsymbol{x}=\left\Vert f\right\Vert _{1}\mathbb{E}\left[\mathrm{sign}(f(\boldsymbol{X}))\right]\label{eq:intf}
\end{equation}
Finally, dividing equation~\eqref{eq:intfg} by equation~\eqref{eq:intf}, using the fact that $\int_{\mathbb{R}^{d}}f(\boldsymbol{x})d\boldsymbol{x}\neq0$, yields equation~\eqref{eq:signed_MC}.
\end{proof}

\begin{rem}
The classical formula~\eqref{eq:EgX}, where $f$ is probability density function, is a special case of~\eqref{eq:signed_MC}.
Indeed, if $f$ happens to be a probability density function, then $\int_{\mathbb{R}^{d}}f(\boldsymbol{x})d\boldsymbol{x}=1$,
$\boldsymbol{X}$ has density $f$, and $\mathrm{sign}(f(\boldsymbol{X}))=1$ almost surely.
\end{rem}

\begin{rem}
The proof of Lemma \ref{lem:signed_MC} shows that equation~\eqref{eq:intfg} provides an alternative representation for $\int_{\mathbb{R}^{d}}f(\boldsymbol{x})g(\boldsymbol{x})d\boldsymbol{x}$ as an expectation. Moreover, unlike equation~\eqref{eq:signed_MC}, equation~\eqref{eq:intfg}
does not require $\int_{\mathbb{R}^{d}}f(\boldsymbol{x})d\boldsymbol{x}\neq0$.
The reason why we favour equation~\eqref{eq:signed_MC} over equation~\eqref{eq:intfg} is that in most applications, such as the random Fourier features application from Section~\ref{sec:signed_rff}, $\int_{\mathbb{R}^{d}}f(\boldsymbol{x})d\boldsymbol{x}$ is known explicitly, while $\left\Vert f\right\Vert _{1}$ has generally no closed-form analytical expression.
\end{rem}

\begin{notation}\label{not:signed_measure}
In the rest of the paper, a function $f\colon\mathbb{R}^{d}\to\mathbb{R}$ as defined in Lemma~\ref{lem:signed_MC} will be called a \emph{signed density}, and we say that a random variable $X$ \emph{has signed density} $f$ if for any measurable function $g\colon\mathbb{R}^{d}\to\mathbb{R}$ such that $fg$ is absolutely integrable, the expectation $\mathbb{E}[g(X)]$ is given by equation~\eqref{eq:signed_MC}. This is a natural way to define the density of a signed measure (see also \citep{monchietti2025measure,polson2025negative}).
\end{notation}

\subsection{Signed Monte Carlo estimator\label{subsec:mc_estimator}}

Lemma \ref{lem:signed_MC} suggests the following signed Monte Carlo
estimator for the integral $\int_{\mathbb{R}^{d}}f(\boldsymbol{x})g(\boldsymbol{x})d\boldsymbol{x}$:
\begin{equation}
\underset{\hspace{0.25em}\mathbb{R}^{d}}{\int}\hspace{-0.125em}f(\boldsymbol{x})g(\boldsymbol{x})d\boldsymbol{x} \ \simeq\ \bigg(\underset{\hspace{0.25em}\mathbb{R}^{d}}{\int}\hspace{-0.125em}f(\boldsymbol{x})d\boldsymbol{x}\bigg)\frac{\sum_{m=1}^{M}\mathrm{sign}(f(\boldsymbol{X}_{m}))g(\boldsymbol{X}_{m})}{\sum_{m=1}^{M}\mathrm{sign}(f(\boldsymbol{X}_{m}))}\label{eq:EgX_SMC}
\end{equation}
where $\boldsymbol{X}_{1}=(X_{1,1},\ldots,X_{d,1})$, $\ldots$, $\boldsymbol{X}_{M}=(X_{1,M},\ldots,X_{d,M})$
are $M$ i.i.d. vector simulations from the distribution of $\boldsymbol{X}=(X_{1},\ldots,X_{d})$.
Indeed, using the strong law of large numbers and Lemma \ref{lem:signed_MC}, the
Monte Carlo estimator~\eqref{eq:EgX_SMC} converges almost surely as $M\to \infty$:
{\scalefont{0.93}
\[
\bigg(\underset{\hspace{0.25em}\mathbb{R}^{d}}{\int}\hspace{-0.125em}f(\boldsymbol{x})d\boldsymbol{x}\bigg)\frac{\sum_{m=1}^{M}\mathrm{sign}(f(\boldsymbol{X}_{m}))g(\boldsymbol{X}_{m})}{\sum_{m=1}^{M}\mathrm{sign}(f(\boldsymbol{X}_{m}))}\stackrel[M\to\infty]{a.s.}{\longrightarrow}\bigg(\underset{\hspace{0.25em}\mathbb{R}^{d}}{\int}\hspace{-0.125em}f(\boldsymbol{x})d\boldsymbol{x}\bigg)\frac{\mathbb{E}\left[\mathrm{sign}(f(\boldsymbol{X}))g(\boldsymbol{X})\right]}{\mathbb{E}\left[\mathrm{sign}(f(\boldsymbol{X}))\right]}=\underset{\hspace{0.25em}\mathbb{R}^{d}}{\int}\hspace{-0.125em}f(\boldsymbol{x})g(\boldsymbol{x})d\boldsymbol{x}.
\]
}

We refer to equation~\eqref{eq:EgX_SMC} as the Signed Monte Carlo (SMC) integration method, because each Monte Carlo simulation $g(\boldsymbol{X}_m)$ is multiplied by either $+1$ or $-1$, depending on the sign of $f(\boldsymbol{X}_m)$. 

\begin{rem}
In practice, a simple way to make the random event $\left\{\sum_{m=1}^{M}\mathrm{sign}(f(\boldsymbol{X}_{m}))=0\right\}$ a.s. impossible in equation~\eqref{eq:EgX_SMC} is to set the number of Monte Carlo samples $M$ as an odd number.
\end{rem}

\section{Signed random Fourier features\label{sec:signed_rff}}

In this section, we extend the spectral Monte Carlo idea underpinning
the random Fourier features \citep{rahimi2007random} technique to
the signed Monte Carlo approach developed in the previous section.
This extension enables us to construct random Fourier features
for kernel functions which are not positive definite. Subsection~\ref{subsec:rff}
recalls the details of the standard random Fourier features (RFF) construction,
and Subsection~\ref{subsec:srff} introduces the aforementioned
extension, called signed random Fourier features (SRFF). Finally,
Subsection~\ref{subsec:special_kernels} simplifies the SRFF construction
in the case of multivariate product and isotropic kernels, which are
of practical importance for kernel density estimation.

\subsection{Random Fourier features (RFF)\label{subsec:rff}}

Let $K\colon\mathbb{R}^{d}\to\mathbb{R}$ be a continuous,
stationary kernel function. $K$ is called positive
definite if for any $N\geq1$, $(\boldsymbol{x}_{1},\ldots,\boldsymbol{x}_{N})\in\mathbb{R}^{d\times N}$
and real coefficients $(z_{1},\ldots,z_{N})\in\mathbb{R}^{N}$,
\begin{equation}
\sum_{i=1}^{N}\sum_{j=1}^{N}z_{i}z_{j}K(\boldsymbol{x}_{i}-\boldsymbol{x}_{j})\geq 0.\label{eq:positive_definite}
\end{equation}
According to Bochner's theorem \citep{bochner1933monotone,bochner1959lectures},
a continuous, stationary kernel $K\colon\mathbb{R}^{d}\to\mathbb{R}$
is positive definite if and only if there exists a finite measure
$\mu$ such that
\[
K(\boldsymbol{u})=\intop_{\mathbb{R}^{d}}\exp(i\boldsymbol{x}^{\top}\boldsymbol{u})d\mu(\boldsymbol{x}),\ \ \boldsymbol{u}\in\mathbb{R}^{d}.
\]
In other words, $K$ is proportional to a characteristic
function. Suppose that $\mu$ is absolutely continuous with respect
to the Lebesgue measure \citep[Theorem~1.8.16]{sasvari2013characteristic}. Then, there exists a (multivariate) density
$f$ such that
\begin{equation}
K(\boldsymbol{u})=K(\boldsymbol{0})\intop_{\mathbb{R}^{d}}\exp(i\boldsymbol{x}^{\top}\boldsymbol{u})f(\boldsymbol{x})d\boldsymbol{x},\ \ \boldsymbol{u}\in\mathbb{R}^{d}.\label{eq:multivariate_fourier_1}
\end{equation}
where $K(\boldsymbol{0})>0$. In other words, $K/K(\boldsymbol{0})$
is the multivariate Fourier transform of $f$. According
to Bochner's theorem, $f$ is nonnegative if and only
if $K$ is positive definite. This means that $f$
is a probability density function, known as the \textit{spectral density}
of $K$, if and only if $K$ is positive
definite. To further simplify equation~\eqref{eq:multivariate_fourier_1}, we
additionally assume that $K$ is symmetric.
\begin{assumption}
\label{assu:symmetric}We assume that the kernel function $K$
is symmetric: $K(\boldsymbol{u})=K(-\boldsymbol{u})$
for all $\boldsymbol{u}\in\mathbb{R}^{d}$.
\end{assumption}

Under Assumption \ref{assu:symmetric}, the function $f$
is symmetric as well and the imaginary part in equation~\eqref{eq:multivariate_fourier_1}
disappears, meaning that equation~\eqref{eq:multivariate_fourier_1}
can be further explicited into the following probabilistic representation:
\begin{equation}
K(\boldsymbol{u})=K(\boldsymbol{0})\intop_{\mathbb{R}^{d}}\cos(\boldsymbol{x}^{\top}\boldsymbol{u})f(\boldsymbol{x})d\boldsymbol{x}=K(\boldsymbol{0})\,\mathbb{E}\!\left[\cos(\bm{\eta}^{\top}\boldsymbol{u})\right],\ \ \boldsymbol{u}\in\mathbb{R}^{d}\label{eq:Ecos_eta_u}
\end{equation}
where $\bm{\eta}=(\eta_{1},\ldots,\eta_{d})$ is a continuous random
vector with density $f$, where $f$ is the Fourier dual of $K$:
\begin{align}
f(\boldsymbol{x}) & =\frac{1}{K(\boldsymbol{0})(2\pi)^{d}}\int_{\mathbb{R}^{d}}\cos(\boldsymbol{x}^{\top}\boldsymbol{u})K(\boldsymbol{u})d\boldsymbol{u},\ \ \boldsymbol{x}\in\mathbb{R}^{d}.\label{eq:f_wrt_K}
\end{align}
The random vector $\bm{\eta}$
is known as \textit{random projection}. To sum up, under the symmetry assumption \ref{assu:symmetric}, the
probabilistic representation~\eqref{eq:Ecos_eta_u} holds if and only
if $K$ is positive definite. In such a case, the multivariate
kernel $K$ can be approximated by Monte Carlo simulations:
\begin{equation}
K(\boldsymbol{u})=K(\boldsymbol{0})\,\mathbb{E}\!\left[\cos(\bm{\eta}^{\top}\boldsymbol{u})\right]\simeq\frac{K(\boldsymbol{0})}{M}\sum_{m=1}^{M}\cos(\bm{\eta}_{m}^{\top}\boldsymbol{u})\ ,\ \boldsymbol{u}\in\mathbb{R}^{d}\label{eq:random_fourier_features}
\end{equation}
where $\bm{\eta}_{1}$, $\ldots$, $\bm{\eta}_{M}$ are $M$ i.i.d.
vector simulations from the distribution of $\bm{\eta}$. This Monte
Carlo approach~\eqref{eq:random_fourier_features} is known as \textit{random Fourier features} \citep{rahimi2007random}, and is a popular way to speed up the computation of kernel-based estimators in machine learning applications. We refer the reader to \citep{langrene2025mixture} for further examples and details on the implementation of random Fourier features.

\subsection{Signed random Fourier features (SRFF)\label{subsec:srff}}

One key limitation of the approach~\eqref{eq:Ecos_eta_u}-\eqref{eq:random_fourier_features}
is that it does not hold if the kernel function is not positive definite.
In this subsection, we explicitly generalize equation~\eqref{eq:random_fourier_features}
to the case where the kernel $K$ is not positive definite,
using the signed Monte Carlo theory developed in Section~\ref{sec:signed_monte_carlo}.
\begin{thm}
\label{thm:SRFF}{[}Signed random Fourier features{]} 
Let $K\colon\mathbb{R}^{d}\to\mathbb{R}$
be a continuous, symmetric kernel function such that $K(\boldsymbol{0})\neq0$,
and such that the inverse Fourier transform $f$ of $K/K(\boldsymbol{0})$ (equation
\eqref{eq:f_wrt_K}) is absolutely integrable. Then, the following
holds:
\begin{equation}
K(\boldsymbol{u})=K(\boldsymbol{0})\frac{\mathbb{E}\left[\mathrm{sign}(f(\bm{\eta}))\cos(\bm{\eta}^{\top}\boldsymbol{u})\right]}{\mathbb{E}\left[\mathrm{sign}(f(\bm{\eta}))\right]}\ ,\ \boldsymbol{u}\in\mathbb{R}^{d}\label{eq:SRFF}
\end{equation}
where the random projection vector $\bm{\eta}=(\eta_{1},\ldots,\eta_{d})$
has probability density function $\overset{_{\,\frown}}{f}=|f|/\left\Vert f\right\Vert _{1}$.
\end{thm}

\begin{proof}
Using equation~\eqref{eq:Ecos_eta_u} with $\boldsymbol{u}=\boldsymbol{0}$ and the condition $K(\boldsymbol{0})\neq0$,
we obtain the following
\[
\intop_{\mathbb{R}^{d}}f(\boldsymbol{x})d\boldsymbol{x}=1\neq0\ ,
\]
which makes Lemma \ref{lem:signed_MC} applicable. Then, using equation~\eqref{eq:Ecos_eta_u} and Lemma~\ref{lem:signed_MC} with $g(\boldsymbol{x})=\cos(\boldsymbol{x}^{\top}\boldsymbol{u})$ for fixed $\boldsymbol{u}\in \mathbb{R}^d$ yields
\begin{align*}
K(\boldsymbol{u}) & =K(\boldsymbol{0})\intop_{\mathbb{R}^{d}}\cos(\boldsymbol{x}^{\top}\boldsymbol{u})f(\boldsymbol{x})d\boldsymbol{x}=K(\boldsymbol{0})\frac{\mathbb{E}\left[\mathrm{sign}(f(\bm{\eta}))\cos(\bm{\eta}^{\top}\boldsymbol{u})\right]}{\mathbb{E}\left[\mathrm{sign}(f(\bm{\eta}))\right]}
\end{align*}
where the random vector $\bm{\eta}=(\eta_{1},\ldots,\eta_{d})$ has density $\overset{_{\,\frown}}{f}=|f|/\left\Vert f\right\Vert _{1}$, which is well defined since $f$ is assumed to be absolutely integrable.
\end{proof}
Theorem \ref{thm:SRFF} provides a probabilistic representation of
continuous, symmetric kernel functions $K$ which does not require $K$ to be positive definite. Indeed, Theorem
\ref{thm:SRFF} does not require the inverse Fourier transform of
$K$ to be nonnegative, it only requires it to be absolutely
integrable \citep{berry1931necessary}. As such, it can be viewed as a generalization of the random
Fourier features approach \citep{rahimi2007random}. We call this extension \emph{signed random Fourier features}.

Theorem \ref{thm:SRFF} suggests the following Monte Carlo estimator
of $K(\boldsymbol{u})$ for any $\boldsymbol{u}\in\mathbb{R}^{d}$:
\begin{equation}
K_{M}(\boldsymbol{u}):=K(\boldsymbol{0})\frac{\sum_{m=1}^{M}\mathrm{sign}(f(\bm{\eta}_{m}))\cos(\bm{\eta}_{m}^{\top}\bm{u})}{\sum_{m=1}^{M}\mathrm{sign}(f(\bm{\eta}_{m}))}\label{eq:SRFF_MC}
\end{equation}
where $\bm{\eta}_{1}=(\eta_{1,1},\ldots\eta_{d,1}),\ldots,\bm{\eta}_{M}=(\eta_{1,M},\ldots\eta_{d,M})$
are $M$ i.i.d. simulations from the random vector $\bm{\eta}=(\eta_{1},\ldots,\eta_{d})$
with density $\overset{_{\,\frown}}{f}$.
This corresponds to the complex feature mapping
\begin{equation}
K_{M}(\boldsymbol{x}_{i}-\boldsymbol{x}_{j})=\varphi_{M}(\boldsymbol{x}_{i})^{\top}\varphi_{M}(\boldsymbol{x}_{j})\label{eq:kernel_dot_product}
\end{equation}
for all $\boldsymbol{x}_{i}\in\mathbb{R}^{d}$, $\boldsymbol{x}_{j}\in\mathbb{R}^{d}$,
where, for any $\boldsymbol{u}\in\mathbb{R}^{d}$, 
\begin{equation}
\varphi_{M}(\boldsymbol{u}):=\frac{\sqrt{K(\boldsymbol{0})}}{\sqrt{\sum_{m=1}^{M}\mathrm{sign}(f(\bm{\eta}_{m}))}}\left[\begin{array}{c}
\sqrt{\mathrm{sign}(f(\bm{\eta}_{1}))}\cos(\bm{\eta}_{1}^{\top}\boldsymbol{u})\\
\vdots\\
\sqrt{\mathrm{sign}(f(\bm{\eta}_{M}))}\cos(\bm{\eta}_{M}^{\top}\boldsymbol{u})\\
\sqrt{\mathrm{sign}(f(\bm{\eta}_{1}))}\sin(\bm{\eta}_{1}^{\top}\boldsymbol{u})\\
\vdots\\
\sqrt{\mathrm{sign}(f(\bm{\eta}_{M}))}\sin(\bm{\eta}_{M}^{\top}\boldsymbol{u})
\end{array}\right]\in\mathbb{C}^{2M}\label{eq:feature_mapping}
\end{equation}
In practice, a convenient way to simulate the random projection $\bm{\eta}$
is the acceptance-rejection method \citep{vonneumann1951random}, as it avoids the computation of the scaling constant $\left\Vert f\right\Vert _{1}$ which appears in the density of $\bm{\eta}$ and is usually not analytical and hard to estimate (see Subsection~\ref{subsec:rejection_sampling}).

\subsection{Special multivariate kernels\label{subsec:special_kernels}}

The signed random Fourier feature mapping \eqref{eq:feature_mapping}
can be readily applied to decompose multivariate indefinite kernels.
That being said, its implementation can be further simplified in the
cases of the most popular multivariate kernel constructions in density
estimation, namely product kernels and isotropic kernels \citep{hardle2004nonparametric}.

\subsubsection{Product kernels\label{subsec:product_kernels}}

Product kernels, also known as multiplicative kernels, are defined as
\begin{equation}
K(\boldsymbol{u})=\prod^{d}_{\ell=1}k_{\ell}(u_{\ell})\label{eq:product_kernel}
\end{equation}
where $\boldsymbol{u}=(u_{1},\ldots,u_{d})\in\mathbb{R}^{d}$, and
$k_{\ell}$, $\ell=1,\ldots,d$, are univariate kernel functions,
usually, but not necessarily, chosen equal. For every dimension $\ell=1,\ldots,d$,
let $f_{\ell}$ be the inverse Fourier transform of $k_{\ell}/k_{\ell}(0)$.

When the product kernel \eqref{eq:product_kernel} is positive definite,
all the univariate kernels $k_{\ell}$, $\ell=1,\ldots,d$, are positive
definite as well, with spectral density $f_{\ell}$. Then, using \citep[Theorem~1.3.10]{sasvari2013characteristic},
equation \eqref{eq:product_kernel} holds if and only if $K$ is the
characteristic function of the random vector $\bm{\eta}=(\eta_{1},\ldots,\eta_{d})$
where the $\eta_{\ell}$ are independent random variables with density
$f_{\ell}$. In particular, the inverse Fourier transform $f$ of
$K/K(\boldsymbol{0})$ satisfies $f(\boldsymbol{x})=\prod^{d}_{\ell=1}f_{\ell}(x_{\ell})$,
for all $\boldsymbol{x}=(x_{1},\ldots,x_{d})\in\mathbb{R}^{d}$.

Consider now the indefinite setting. When the product kernel \eqref{eq:product_kernel} is not positive
definite, some of the functions $f_{\ell}$ are not densities anymore, and may take negative values. Suppose that each $f_{\ell}$ is nevertheless absolutely
integrable. Then, the equation $f(\boldsymbol{x})=\prod^{d}_{\ell=1}f_{\ell}(x_{\ell})$
still holds, and the product kernel \eqref{eq:product_kernel} admits
the signed random Fourier features representation \eqref{eq:SRFF}
with the random projection vector $\bm{\eta}=(\eta_{1},\ldots,\eta_{d})$
defined such that the random variables $\eta_{\ell}$ are independent
with density $\left|f_{\ell}\right|/\left\Vert f_{\ell}\right\Vert _{1}$
respectively.

\subsubsection{Isotropic kernels\label{subsec:isotropic_kernels}}

Isotropic kernels, also known as radially symmetric kernels, are defined
as
\begin{equation}
K(\boldsymbol{u})=k(\left\Vert \boldsymbol{u}\right\Vert )\label{eq:isotropic_kernel}
\end{equation}
where $\boldsymbol{u}\in\mathbb{R}^{d}$, and $k\colon \mathbb{R}^+\to \mathbb{R}$ is a univariate kernel function.

When the isotropic kernel \eqref{eq:isotropic_kernel} is positive
definite, its spectral density $f$ exists and is also isotropic.
Consequently, the spectral distribution $f$ admits an explicit representation
as a scale mixture of the random distribution on the $d$-dimensional
unit sphere, as proved in the following proposition.
\begin{prop}
\label{prop:scale_mixture}The $d$-dimensional random vector $\boldsymbol{X}$
has a radial (a.k.a. isotropic) density $f\colon\mathbb{R}^{d}\to\mathbb{R}$,
where $f(\boldsymbol{x})=\mathfrak{f}(\left\Vert \boldsymbol{x}\right\Vert )$
with $\mathfrak{f}:\mathbb{R}^{+}\to\mathbb{R}^{+}$, if and only if
there exists a random vector $\boldsymbol{U}$ uniformly
distributed on the $d$-dimensional unit sphere, and a nonnegative random
variable $R$, independent of $\boldsymbol{U}$, whose density $f_{\:\!\!R}:\mathbb{R}^{+}\to\mathbb{R}^{+}$
is explicitly given by 
\begin{equation}
f_{\:\!\!R}(r)=\frac{2\pi^{\frac{d}{2}}}{\Gamma\!\left(\frac{d}{2}\right)}r^{d-1}\mathfrak{f}(r)\mathbbm{1}_{\{r\geq0\}}\ ,\ \forall r\geq0,\label{eq:radius_density}
\end{equation}
such that $\boldsymbol{X}\overset{d}{=}R\boldsymbol{U}$.
\end{prop}

\begin{proof}
The equivalence between the isotropy of the distribution of $\boldsymbol{X}$
and its scale mixture representation $\boldsymbol{X}\overset{d}{=}R\boldsymbol{U}$
is a consequence of Schoenberg's theorem \citep[Corollary~3.8.3]{sasvari2013characteristic}.
Next, since $f$ is a density, $\int_{\mathbb{R}^{d}}f(\boldsymbol{x})d\boldsymbol{x}=1$.
Using a polar change of variable \citep[Corollary 2.51 and Proposition 2.54 page 79]{folland1999real},
\[
\int_{\mathbb{R}^{d}}f(\boldsymbol{x})d\boldsymbol{x}=1=\int_{\mathbb{R}^{d}}\mathfrak{f}(\left\Vert \boldsymbol{x}\right\Vert )d\boldsymbol{x}=\frac{2\pi^{\frac{d}{2}}}{\Gamma\!\left(\frac{d}{2}\right)}\int^{\infty}_{0}r^{d-1}\mathfrak{f}(r)dr
\]
This ensures that the nonnegative function $h(r):=\frac{2\pi^{\frac{d}{2}}}{\Gamma(\frac{d}{2})}r^{d-1}\mathfrak{f}(r)\mathbbm{1}_{\{r\geq0\}}$
is integrable, and integrates to one. In other words, $h$ is a probability
density function. Then, for any $t\geq0$, using the same polar change
of variable, 
\[
\mathbb{P}\left(\left\Vert \boldsymbol{X}\right\Vert \leq t\right)=\int_{\mathbb{R}^{d}}\mathbbm{1}_{\{\left\Vert \boldsymbol{x}\right\Vert \leq t\}}\mathfrak{f}(\left\Vert \boldsymbol{x}\right\Vert )d\boldsymbol{x}=\frac{2\pi^{\frac{d}{2}}}{\Gamma\!\left(\frac{d}{2}\right)}\int^{\infty}_{0}r^{d-1}\mathfrak{f}(r)\mathbbm{1}_{\{r\leq t\}}dr
\]
On the other hand, $\left\Vert \boldsymbol{X}\right\Vert =R\left\Vert \boldsymbol{U}\right\Vert =R$
and
\[
\mathbb{P}\left(\left\Vert \boldsymbol{X}\right\Vert \leq t\right)=\mathbb{P}\left(R\leq t\right)=\int^{\infty}_{0}f_{\:\!\!R}(r)\mathbbm{1}_{\{r\leq t\}}dr.
\]
Since cumulative distribution functions fully characterize probability
distributions, this proves that $f_{\:\!\!R}(r)=h(r)=\frac{2\pi^{\frac{d}{2}}}{\Gamma(\frac{d}{2})}r^{d-1}\mathfrak{f}(r)\mathbbm{1}_{\{r\geq0\}}$
for all $r\geq0$. In other words, equation \eqref{eq:radius_density}
holds and defines a probability density function.
\end{proof}

\begin{rem}
Let $\Phi_{d}$ denote the set of kernels that are positive definite
in $\mathbb{R}^{d}$, and let $\Phi_{\infty}$ denote the set of kernels that
are positive definite in $\mathbb{R}^{d}$ for all $d\geq1$. Recall
that $\Phi_{\infty}\subset\cdots\subset\Phi_{2}\subset\Phi_{1}$ \citep[page~173]{sasvari2013characteristic}.
The spectral scale mixture described in Proposition~\ref{prop:scale_mixture}
is valid for every kernel in $\Phi_{d}$. It is therefore more general
than the Gaussian spectral scale mixture described in \citep{langrene2025mixture},
which only applies to kernels in $\Phi_{\infty}$. Examples of kernels
in $\Phi_{1}\setminus\Phi_{\infty}$ include the triangular kernel
$K(u)=(1-\left|u\right|)\mathbbm{1}_{\{\left|u\right|\leq1\}}$, whose
spectral density is the Fej\'er-de la Vall\'ee-Poussin density $f(x)=\frac{1-\cos(x)}{\pi x^{2}}=\frac{1}{2\pi}\frac{\sin^{2}(x/2)}{(x/2)^{2}}$
\citep{davis1975mean}\citep[page 207]{hormann2004automatic}\citep[page 100]{devroye2006nonuniform},
and the Silverman kernel $K(u)=\frac{1}{2}\exp\left(-\frac{\left|u\right|}{\sqrt{2}}\right)\sin\left(\frac{\left|u\right|}{\sqrt{2}}+\frac{\pi}{4}\right)$ \citep{silverman1984spline,tsybakov2009nonparametric}, whose spectral density is the Laha density $f(x)=\frac{\sqrt{2}}{\pi(1+x^{4})}$
\citep{laha1958example,crooks2019field}, a.k.a. Butterworth filter
\citep{hardle1994kernel}. However, unlike the construction described
in \citep{langrene2025mixture}, equation~\eqref{eq:radius_density}
requires an analytical expression for the spectral density $f(.)=\mathfrak{f}(\left\Vert .\right\Vert )$
of the kernel under consideration.
\end{rem}

When the isotropic kernel \eqref{eq:isotropic_kernel} is not positive
definite, equation \eqref{eq:radius_density} does not define a density
anymore, and $f_{\:\!\!R}$ may take negative values. Suppose that $f_{\:\!\!R}$
is nevertheless absolutely integrable. Then, the isotropic kernel
\eqref{eq:isotropic_kernel} admits the signed random Fourier features
representation \eqref{eq:SRFF} with the random projection vector
$\bm{\eta}=R\boldsymbol{U}$, where $\boldsymbol{U}$ is a random
vector uniformly distributed on the $d$-dimensional unit sphere, and
$R$ is a random variable, independent of $\boldsymbol{U}$, with
density $\left|f_{\:\!\!R}\right|/\left\Vert f_{\:\!\!R}\right\Vert_1 $,
where $f_{\:\!\!R}$ is defined by equation \eqref{eq:radius_density}.

\section{Fast kernel density estimation\label{sec:fast_kde}}

This section describes how to use the signed random Fourier features methodology (Section~\ref{sec:signed_rff}) to improve the computational efficiency of kernel density estimation. 

\subsection{Fast KDE computation\label{subsec:fast_kde_computation}}

Let $\mathcal{D}=\left\{ (\boldsymbol{x}_{n},y_{n})\in\mathbb{R}^{d}\times\mathbb{R}\ensuremath{,}n=1,\ldots,N\right\} $
be a dataset and $K:\mathbb{R}^{d}\to\mathbb{R}$
be a stationary, symmetric kernel function. Define the column vector
$\boldsymbol{y}=\left[y_{n}\right]_{1\leq n\leq N}\in\mathbb{R}^{N}$
and the kernel matrix $\mathbf{K}:=\left[K(\boldsymbol{x}_{i}-\boldsymbol{x}_{j})\right]_{1\leq i,j\leq N}\in\mathbb{R}^{N\times N}$.
We are interested in computing the following kernel matrix vector
multiplication (MVM):
\begin{equation}
\mathbf{K}\boldsymbol{y}=\left[\begin{array}{c}
\sum_{n=1}^{N}y_{n}K(\boldsymbol{x}_{1}-\boldsymbol{x}_{n})\\
\sum_{n=1}^{N}y_{n}K(\boldsymbol{x}_{2}-\boldsymbol{x}_{n})\\
\vdots\\
\sum_{n=1}^{N}y_{n}K(\boldsymbol{x}_{N}-\boldsymbol{x}_{n})
\end{array}\right]\label{eq:kernel_mvm}
\end{equation}

Equation~\eqref{eq:kernel_mvm} contains the kernel density estimator $\mathbf{K}\boldsymbol{1}/N$, where $\boldsymbol{1}:=[1,\ldots,1]^{\top}\in\mathbb{R}^{N}$, as a particular case, and the Nadaraya-Watson kernel regression estimator can be written as $\mathbf{K}\boldsymbol{y}/\mathbf{K}\boldsymbol{1}$. Both estimators are evaluated at each data point $\boldsymbol{x}_{1},\ldots,\boldsymbol{x}_{N}$.

A direct computation of the kernel MVM~\eqref{eq:kernel_mvm} requires
$\mathcal{O}(N^{2})$ operations. However, by using the signed random
Fourier features approximation~\eqref{eq:SRFF_MC} of the kernel function
$K$ with $M$ random projections $\bm{\eta}_{1}$, $\ldots$,
$\bm{\eta}_{M}$, one obtains for any $\boldsymbol{z}\in\mathbb{R}^{d}$,

\begin{align*}
 & \sum_{n=1}^{N}y_{n}K(\boldsymbol{z}-\boldsymbol{x}_{n})\simeq K(\boldsymbol{0})\sum_{n=1}^{N}y_{n}\frac{\sum_{m=1}^{M}\mathrm{sign}(f(\bm{\eta}_{m}))\times\cos(\bm{\eta}_{m}^{\top}(\boldsymbol{z}-\boldsymbol{x}_{n}))}{\sum_{m=1}^{M}\mathrm{sign}(f(\bm{\eta}_{m}))}\\
 & =K(\boldsymbol{0})\frac{\sum_{n=1}^{N}y_{n}\sum_{m=1}^{M}\mathrm{sign}(f(\bm{\eta}_{m}))\times\left(\cos(\bm{\eta}_{m}^{\top}\boldsymbol{z})\cos(\bm{\eta}_{m}^{\top}\boldsymbol{x}_{n})+\sin(\bm{\eta}_{m}^{\top}\boldsymbol{z})\sin(\bm{\eta}_{m}^{\top}\boldsymbol{x}_{n})\right)}{\sum_{m=1}^{M}\mathrm{sign}(f(\bm{\eta}_{m}))}\\
 & =K(\boldsymbol{0})\left(\frac{\sum_{m=1}^{M}\mathrm{sign}(f(\bm{\eta}_{m}))\cos(\bm{\eta}_{m}^{\top}\boldsymbol{z})\left\{ \sum_{n=1}^{N}y_{n}\cos(\bm{\eta}_{m}^{\top}\boldsymbol{x}_{n})\right\} }{\sum_{m=1}^{M}\mathrm{sign}(f(\bm{\eta}_{m}))}\right.\\
 & \left.+\frac{\sum_{m=1}^{M}\mathrm{sign}(f(\bm{\eta}_{m}))\sin(\bm{\eta}_{m}^{\top}\boldsymbol{z})\left\{ \sum_{n=1}^{N}y_{n}\sin(\bm{\eta}_{m}^{\top}\boldsymbol{x}_{n})\right\} }{\sum_{m=1}^{M}\mathrm{sign}(f(\bm{\eta}_{m}))}\right) .
\end{align*}
This formula can be computed for every $\boldsymbol{z}\in\left\{\boldsymbol{x}_{1},\boldsymbol{x}_{2},\ldots,\boldsymbol{x}_{N}\right\} $
in $\mathcal{O}(MN)$ operations, as described in Algorithm \ref{algo:fast_kernel_mvm}.

\begin{algorithm2e}[H]
\DontPrintSemicolon 
\SetAlgoLined 

\vspace{1mm}

\KwIn{Dataset $\mathcal{D}=\left\{ (\boldsymbol{x}_{n},y_{n})\in\mathbb{R}^{d}\times\mathbb{R}\ensuremath{,}\ n=1,\ldots,N\right\} $,
\newline Stationary, symmetric kernel $K:\mathbb{R}^{d}\to\mathbb{R}$,
\vspace{1mm}\newline $f$, the inverse Fourier transform of $K/K(\boldsymbol{0})$,
\vspace{1mm}\newline Random projections $\bm{\eta}_{m}\in\mathbb{R}^{d}$,
$m=1,\ldots,M$ drawn independently from the density $\left|f\right|/\left\Vert f\right\Vert _{1}$.}

\vspace{1mm}

\For{$m=1,...,M$}{

\vspace{0.5mm}

$\mathrm{\delta_{m}:=sign}(f(\bm{\eta}_{m}))$

\vspace{0.5mm}

\For{$n=1,...,N$}{

\vspace{0.5mm}

$c_{m,n}:=\cos(\bm{\eta}_{m}^{\top}\boldsymbol{x}_{n})$

$s_{m,n}:=\sin(\bm{\eta}_{m}^{\top}\boldsymbol{x}_{n})$

} 

} 

$C_{m}:=\sum^{N}_{n=1}y_{n}c_{m,n}$

\vspace{1mm}

$S_{m}:=\sum^{N}_{n=1}y_{n}s_{m,n}$

\vspace{1mm}

$\Delta:=\sum^{M}_{m=1}\mathrm{sign}(f(\bm{\eta}_{m}))$

\vspace{1mm}

\For{$n=1,...,N$}{

\vspace{0.5mm}

$R_{n}:=\frac{K(\boldsymbol{0})}{\Delta}\sum^{M}_{m=1}\left(\delta_{m}C_{m}c_{m,n}+\delta_{m}S_{m}s_{m,n}\right)$

} 

\KwOut{$R_{n}$, $n=1,2,\ldots,N$}

\vspace{1mm}

\Comment*[l]{$R_{n}=K(\boldsymbol{0})\sum^{N}_{n'=1}y_{n'}\frac{\sum^{M}_{m=1}\mathrm{sign}(f(\bm{\eta}_{m}))\cos(\bm{\eta}_{m}^{\top}(\boldsymbol{x}_{n}-\boldsymbol{x}_{n'}))}{\sum^{M}_{m=1}\mathrm{sign}(f(\bm{\eta}_{m}))}\simeq\sum^{N}_{n'=1}y_{n'}K(\boldsymbol{x}_{n}-\boldsymbol{x}_{n'})$}

\vspace{1mm}

\caption{Fast kernel MVM by signed random Fourier features\label{algo:fast_kernel_mvm}}

\end{algorithm2e}

Observe that Algorithm~\ref{algo:fast_kernel_mvm} does not directly use the complex-valued feature mapping~\eqref{eq:feature_mapping}, and therefore does not require to work with complex numbers.

\begin{rem}
As discussed in \citep{langrene2019fast}, kernel density
estimators with adaptive bandwidths, such as $\frac{1}{Nh(\boldsymbol{z})^{d}}\sum^{N}_{n=1}K(\frac{\boldsymbol{z}-\boldsymbol{x}_{n}}{h(\boldsymbol{z})})$
(balloon bandwidth \citep{terrell1992variable}, where the bandwidths
$h(\boldsymbol{z})>0$ depend on the evaluation point $\boldsymbol{z}\in\mathbb{R}^{d}$),
do not benefit from the decomposition provided by the cosine formula
$\cos\!\left(\frac{\bm{\eta}^{\top}_{m}(\boldsymbol{z}-\boldsymbol{x}_{n})}{h(\boldsymbol{z})}\right)=\cos\!\left(\frac{\bm{\eta}^{\top}_{m}\boldsymbol{z}}{h(\boldsymbol{z})}\right)\cos\!\left(\frac{\bm{\eta}^{\top}_{m}\boldsymbol{x}_{n}}{h(\boldsymbol{z})}\right)+\sin\!\left(\frac{\bm{\eta}^{\top}_{m}\boldsymbol{z}}{h(\boldsymbol{z})}\right)\sin\!\left(\frac{\bm{\eta}^{\top}_{m}\boldsymbol{x}_{n}}{h(\boldsymbol{z})}\right)$,
as the resulting cross terms $\cos\!\left(\frac{\bm{\eta}^{\top}_{m}\boldsymbol{x}_{n}}{h(\boldsymbol{z})}\right)$
and $\sin\!\left(\frac{\bm{\eta}^{\top}_{m}\boldsymbol{x}_{n}}{h(\boldsymbol{z})}\right)$
cannot disentangle the sources $\boldsymbol{x}$ from the targets
$\boldsymbol{z}$, making fast summation approaches such as Algorithm~\ref{algo:fast_kernel_mvm}
inapplicable. Instead, one can mimic the $k$-NN balloon bandwidth
estimator with equal number of data points in each bandwidth support
\citep{loftsgaarden1965nonparametric,terrell1992variable,langrene2019fast},
known to be robust in multivariate settings \citep{terrell1992variable},
by applying a uniformization feature mapping $\boldsymbol{x}_{i}\mapsto\mathrm{rank}(\boldsymbol{x}_{i})/N$
to the input dataset $\mathcal{X}=\left\{ \boldsymbol{x}_{n}\in\mathbb{R}^{d}\ensuremath{,}n=1,\ldots,N\right\} $,
where $\mathrm{rank}(\boldsymbol{x}_{i})$ is the rank of $\boldsymbol{x}_{i}$
in $\mathcal{X}$, and then applying the fast algorithm~\ref{algo:fast_kernel_mvm}
with constant bandwidth to this transformed dataset.
\end{rem}

\subsection{Choice of multivariate kernel function\label{subsec:choice_of_kernel}}

The implementation of the fast KDE algorithm described in the above
Subsection~\ref{subsec:fast_kde_computation} depends on the choice
of kernel function $K$, as spectral sampling requires to know the
Fourier transform of $K$.

Only a few kernel functions commonly used in KDE, such as the multivariate
Gaussian kernel $K(\boldsymbol{u})=\exp\!\left(-\frac{1}{2}\left\Vert \boldsymbol{u}\right\Vert ^{2}\right)$,
are positive definite in $\mathbb{R}^{d}$ for all $d\geq1$. For
such kernels, KDE computations can be performed efficiently using
the classical random Fourier features approach~\eqref{eq:random_fourier_features}.
We refer to \citep{langrene2025mixture} for additional examples of
such kernels, along with their corresponding spectral sampling schemes.

However, the vast majority of kernel functions commonly used in KDE
are either only positive definite for low-dimensional datasets, or
are simply not positive definite. Many of them are particular cases
of the following multivariate isotropic compact kernel, which may
be called the Kuttner-Golubov kernel, or central generalized beta
kernel, defined by
\begin{equation}
K(\boldsymbol{u})=k(\left\Vert \boldsymbol{u}\right\Vert )=(1-\left\Vert \boldsymbol{u}\right\Vert ^{\alpha})^{\beta}\mathbbm{1}_{\{\left\Vert \boldsymbol{u}\right\Vert \leq1\}}\ ,\ \boldsymbol{u}\in\mathbb{R}^{d},\label{eq:kuttner_golubov_kernel}
\end{equation}
where $\alpha>0$ and $\beta>0$. Finding the conditions on $\alpha$
and $\beta$ such that $K$ is positive definite in $\mathbb{R}^{d}$
is the so-called Kuttner-Golubov problem \citep{kuttner1944riesz,golubov1981abel,gneiting2001polya},
and remains an open problem in the general case. The class of compact
kernels \eqref{eq:kuttner_golubov_kernel} encompasses, up to a multiplicative
constant (see Lemma~\ref{lem:kernel_integral}), many kernels of
interest in density estimation and kernel smoothing \citep{silverman1986density,loader1999local,scott2015multivariate,chacon2018multivariate},
including the triangular kernel ($\alpha=1$, $\beta=1$, a.k.a. hat
function \citep{schaback1995radial}), the parabolic kernel ($\alpha=2$,
$\beta=1$,  a.k.a. Epanechnikov kernel \citep{epanechnikov1969nonparametric}), the biweight
kernel ($\alpha=2$, $\beta=2$, a.k.a. quartic, or bisquare kernel \citep{loader1999local}),
the triweight kernel ($\alpha=2$, $\beta=3$), and the tricube kernel
($\alpha=3$, $\beta=3$). These are summarized in Table~\ref{tab:kuttner_golubov}.

\begin{table}[th]
\begin{centering}
\begin{tabular}{lccl}
\hline 
Kernel name & $\alpha$ & $\beta$ & Formula\tabularnewline
\hline 
Triangular & $1$ & $1$ & $K(\boldsymbol{u})=(1-\left\Vert \boldsymbol{u}\right\Vert )\mathbbm{1}_{\{\left\Vert \boldsymbol{u}\right\Vert \leq1\}}$\tabularnewline
Parabolic & $2$ & $1$ & $K(\boldsymbol{u})=(1-\left\Vert \boldsymbol{u}\right\Vert ^{2})\mathbbm{1}_{\{\left\Vert \boldsymbol{u}\right\Vert \leq1\}}$\tabularnewline
Biweight & $2$ & $2$ & $K(\boldsymbol{u})=(1-\left\Vert \boldsymbol{u}\right\Vert ^{2})^{2}\mathbbm{1}_{\{\left\Vert \boldsymbol{u}\right\Vert \leq1\}}$\tabularnewline
Triweight & $2$ & $3$ & $K(\boldsymbol{u})=(1-\left\Vert \boldsymbol{u}\right\Vert ^{2})^{3}\mathbbm{1}_{\{\left\Vert \boldsymbol{u}\right\Vert \leq1\}}$\tabularnewline
Tricube & $3$ & $3$ & $K(\boldsymbol{u})=(1-\left\Vert \boldsymbol{u}\right\Vert ^{3})^{3}\mathbbm{1}_{\{\left\Vert \boldsymbol{u}\right\Vert \leq1\}}$\tabularnewline
\hline 
\end{tabular}
\par\end{centering}
\caption{Examples of kernel functions included in the Kuttner-Golubov class~\eqref{eq:kuttner_golubov_kernel}\label{tab:kuttner_golubov}}
\end{table}

\begin{rem}
Here, the kernel function~\eqref{eq:kuttner_golubov_kernel} is defined
such that $K(\boldsymbol{0})=1$ for simplicity, however smoothing
tasks may require that $\int_{\mathbb{R}^{d}}K(\boldsymbol{u})d\boldsymbol{u}=1$.
This latter scaling can be achieved by dividing equation~\eqref{eq:kuttner_golubov_kernel}
by the constant term computed in Lemma~\ref{lem:kernel_integral}
in Appendix.
\end{rem}

Proposition~\ref{prop:fourier_kuttner_golubov} in Appendix provides
an explicit formula for the Fourier transform of the general Kuttner-Golubov
kernel~\eqref{eq:kuttner_golubov_kernel} in terms of the Fox-Wright
generalized hypergeometric function, as well as more explicit Fourier
transform formulas in the special cases $\alpha=1$ (Askey kernels,
Corollary~\ref{cor:fourier_askey}) and $\alpha=2$ (symmetric beta
kernels, Corollary~\ref{cor:fourier_symmetric_beta}) in terms of
confluent hypergeometric functions. This provides explicit formulas
for the function $f_{\:\!\!R}$ from equation~\eqref{eq:radius_density}.
One can verify that $f_{\:\!\!R}$ is absolutely integrable if and
only if $d\leq2\beta$, in which case the spectral radius density
$\left|f_{\:\!\!R}\right|/\left\Vert f_{\:\!\!R}\right\Vert _{1}$
is well defined, and can be simulated efficiently by acceptance-rejection,
as shown in the next subsection. If $d>2\beta$, then the isotropic
kernel approach (subsection~\ref{subsec:isotropic_kernels}) is not
compatible with the signed random Fourier features methodology~\eqref{eq:SRFF}
(because $\left\Vert f_{\:\!\!R}\right\Vert _{1}=\infty$). In this case, one should either change kernel and pick one such that $\beta\geq\frac{d}{2}$, or use the product kernel approach (subsection~\ref{subsec:product_kernels}), which is compatible with SRFF without any dimensional restriction.

\subsection{Spectral rejection sampling\label{subsec:rejection_sampling}}

In this subsection, we describe how to simulate isotropic spectral densities,
as described in Subsection~\ref{subsec:isotropic_kernels}, in practice.
This step is required to implement the signed random Fourier features
estimator \eqref{eq:SRFF_MC}, as well as the fast kernel density
estimator described in Algorithm~\ref{algo:fast_kernel_mvm}. In
the case of isotropic kernels, the random projection $\boldsymbol{\eta}$
has the distribution of a scale mixture $R\boldsymbol{U}$. The random
vector $\boldsymbol{U}$ is uniformly distributed on the $d$-dimensional
unit sphere, and can easily be simulated as $\boldsymbol{U}=\frac{\boldsymbol{G}}{\left\Vert \boldsymbol{G}\right\Vert }$,
where $\boldsymbol{G}$ is a standard $d$-dimensional Gaussian vector.
The random variable $R$ is independent of $\boldsymbol{U}$, and
has density $\left|f_{\:\!\!R}\right|/\left\Vert f_{\:\!\!R}\right\Vert _{1}$,
where $f_{\:\!\!R}$ is defined by equation \eqref{eq:radius_density}.
Simulating $R$ is the only difficulty in the simulation of the random
projection $\boldsymbol{\eta}$.

From equation~\eqref{eq:radius_density}, the density of the spectral
radius $R$ is proportional to the function $r^{d-1}\left|\mathfrak{f}(r)\right|$,
$r\geq0$, where $f(\boldsymbol{x})=\mathfrak{f}(\left\Vert \boldsymbol{x}\right\Vert )$,
$\boldsymbol{x}\in\mathbb{R}^{d}$ is the inverse Fourier transform of the
kernel $K$ under consideration. Remark that the function $\mathfrak{f}$
is going to depend on $d$, see for example equations~\eqref{eq:fourier_kuttner_golubov}, \eqref{eq:fourier_symmetric_beta}, and \eqref{eq:fourier_askey}.
In order to simulate from such density, the most practical option
is to use acceptance-rejection sampling \citep{vonneumann1951random}.
To do so, it is sufficient to find an integrable, nonnegative envelope
$g\colon\mathbb{R}^{+}\to\mathbb{R}^{+}$ such that
\begin{equation}
r^{d-1}\left|\mathfrak{f}(r)\right|\leq g(r),\ \forall r\geq0,\label{eq:envelope}
\end{equation}
and such that it is easy to simulate from the density $g/\left\Vert g\right\Vert _{1}$.
Then, the acceptance-rejection algorithm works as follows.

\SetKwInput{KwIn}{Input}\SetKwInput{KwOut}{Output}
\begin{algorithm2e}[H]
\DontPrintSemicolon 
\SetAlgoLined 
\LinesNumbered

\vspace{1mm}

\KwIn{Function $g$ satisfying the inequality~\eqref{eq:envelope}}

Let $r$ be a simulation from the density $g/\left\Vert g\right\Vert _{1}$.

Let $u\sim\mathcal{U}(0,1)$ be a simulation from a standard uniform
distribution, independent of~$r$

\uIf{$u\leq\frac{r^{d-1}\left|\mathfrak{f}(r)\right|}{g(r)}$}{

Return $r$

}\Else{

Go back to line 1\textbf{ }

} 

\KwOut{$r$ is a simulation from the spectral radius $R$.}\vspace{1mm}
\caption{Simulating random spectral radius by acceptance-rejection\label{algo:acceptance_rejection}}

\end{algorithm2e}

In the case of the isotropic compact kernels from Table~\ref{tab:kuttner_golubov},
using a uniform proposal distribution as in \citep{liu2021fast} or
a Gaussian proposal distribution as in \citep{he2024random} would
not work, because the tails of these proposal densities decay much
more quickly than the one of the Fourier transform \eqref{eq:fourier_kuttner_golubov},
which decays as an inverse power function. Instead, we propose the
following envelope function $g$:
\begin{equation}
g(r)=\min(c_{L}r^{d-1},c_{M},c_{R}r^{-1-\nu})\mathbbm{1}_{\{r\geq0\}}\label{eq:proposal_envelope}
\end{equation}
for some positive constants, $c_{L}$, $c_{M}$, $c_{R}$, $\nu$.
First, $\nu$ is set to match the decay of the function $r^{d-1}\left|\mathfrak{f}(r)\right|$
for large $r$. Then, the three constants $c_{L}$, $c_{M}$, and
$c_{R}$ are set as follows:
\[
\begin{array}{ccccc}
c_{L}=f(\boldsymbol{0})=\mathfrak{f}(0) & , & c_{M}=\underset{r\geq0}{\max}\left(r^{d-1}\left|\mathfrak{f}(r)\right|\right) & , & c_{R}=\underset{r\geq0}{\max}\left(r^{\nu+d}\left|\mathfrak{f}(r)\right|\right)\end{array}.
\]
The constant $c_{L}$ can be computed from equation~\eqref{eq:f_wrt_K}
(it is the maximum value of the inverse Fourier transform $f$), while the
constants $c_{M}$ and $c_{R}$ can be estimated numerically for fixed
$d$. The density $g/\left\Vert g\right\Vert _{1}$ is a mixture of
a beta, uniform, and Pareto densities with non-overlapping support,
and can easily be simulated using the composition approach, as explained
in Algorithm~\ref{algo:proposal_sampling} below.

\SetKwInput{KwIn}{Inputs}\SetKwInput{KwOut}{Output}
\begin{algorithm2e}[H]
\DontPrintSemicolon 
\SetAlgoLined 
\LinesNumbered

\vspace{1mm}

\KwIn{Constants $c_{L}$, $c_{M}$, $c_{R}$, and $\nu$ from equation~\eqref{eq:proposal_envelope}}

Define $r_{1}=\left(\frac{c_{M}}{c_{L}}\right)^{\frac{1}{d-1}}$ and
$r_{2}=\left(\frac{c_{R}}{c_{M}}\right)^{\frac{1}{1+\nu}}$

Let $u_{1}\sim\mathcal{U}(0,1)$ and $u_{2}\sim\mathcal{U}(0,1)$
be two independent standard uniform simulations

\vspace{1mm}

\uIf{$u_{1}<\frac{c_{L}r^{d}_{1}/d}{c_{L}r^{d}_{1}/d+c_{M}(r_{2}-r_{1})+c_{R}r^{-\nu}_{2}\nu}$}{

\vspace{1mm}

Return $r_{1}u^{\frac{1}{d}}_{2}$ \Comment*[l]{Beta distribution
on $[0,r_{1}]$}

\vspace{1mm}

}\uElseIf{$u_{1}<\frac{c_{L}r^{d}_{1}/d+c_{M}(r_{2}-r_{1})}{c_{L}r^{d}_{1}/d+c_{M}(r_{2}-r_{1})+c_{R}r^{-\nu}_{2}\nu}$}{

\vspace{1mm}

Return $r_{1}+(r_{2}-r_{1})u_{2}$ \Comment*[l]{Uniform distribution
on $[r_{1},r_{2}]$}

}\Else{

Return $r_{2}u^{-\frac{1}{\nu}}_{2}$\textbf{ }\Comment*[l]{Pareto
distribution on $[r_{2},\infty)$}

} 

\KwOut{One simulation from the spectral proposal density $g/\left\Vert g\right\Vert _{1}$.}\vspace{1mm}
\caption{Simulating proposal distribution $g/\left\Vert g\right\Vert _{1}$
from equation~\eqref{eq:proposal_envelope}\label{algo:proposal_sampling}}

\end{algorithm2e}

Figures~\ref{fig:proposal_triangular}-\ref{fig:proposal_parabolic}-\ref{fig:proposal_triweight}
illustrate the proposed envelope function~\eqref{eq:proposal_envelope}
with optimized parameters in the case of the triangular, parabolic,
and triweight kernels respectively, from $d=1$ up to $d=6$ (in the
case of the triweight kernel).
\begin{rem}
In the univariate case $d=1$, which also covers the case of product
kernels~\eqref{eq:product_kernel}, the beta component of the proposal
distribution~\eqref{eq:proposal_envelope} disappears, and only the
uniform and Pareto components remain. The resulting simplified curve
$g$ is a scaled version of the one proposed in \citep[page~207]{hormann2004automatic}
to simulate the Fej\'er-de la Vall\'ee-Poussin distribution (Figure~\ref{fig:proposal_triangular},
left) by acceptance-rejection.
\end{rem}

\begin{figure}[H]
\begin{minipage}[t]{0.48\columnwidth}%
\includegraphics[width=0.38\paperwidth]{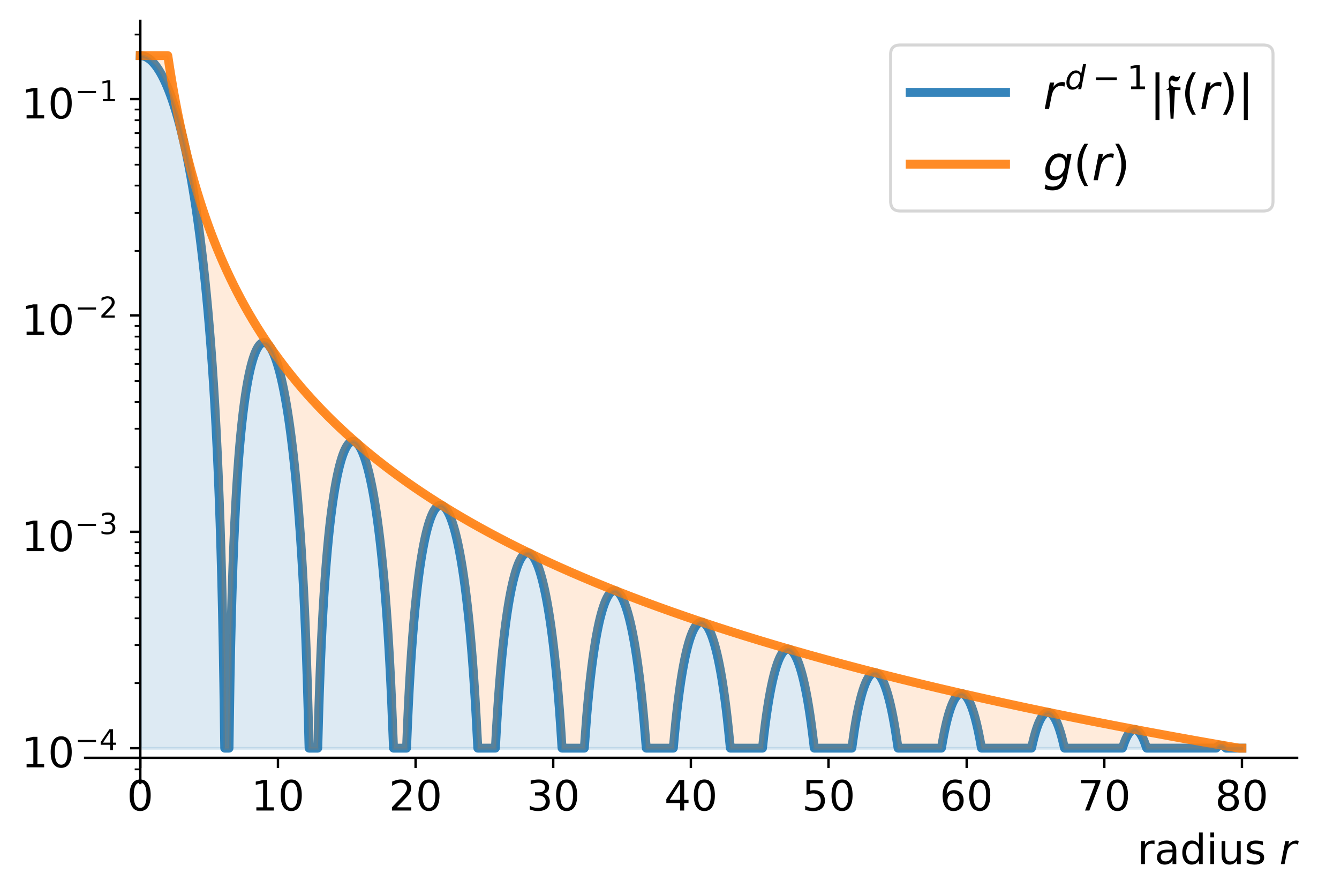}%
\end{minipage}\hfill{}%
\begin{minipage}[t]{0.48\columnwidth}%
\includegraphics[width=0.38\paperwidth]{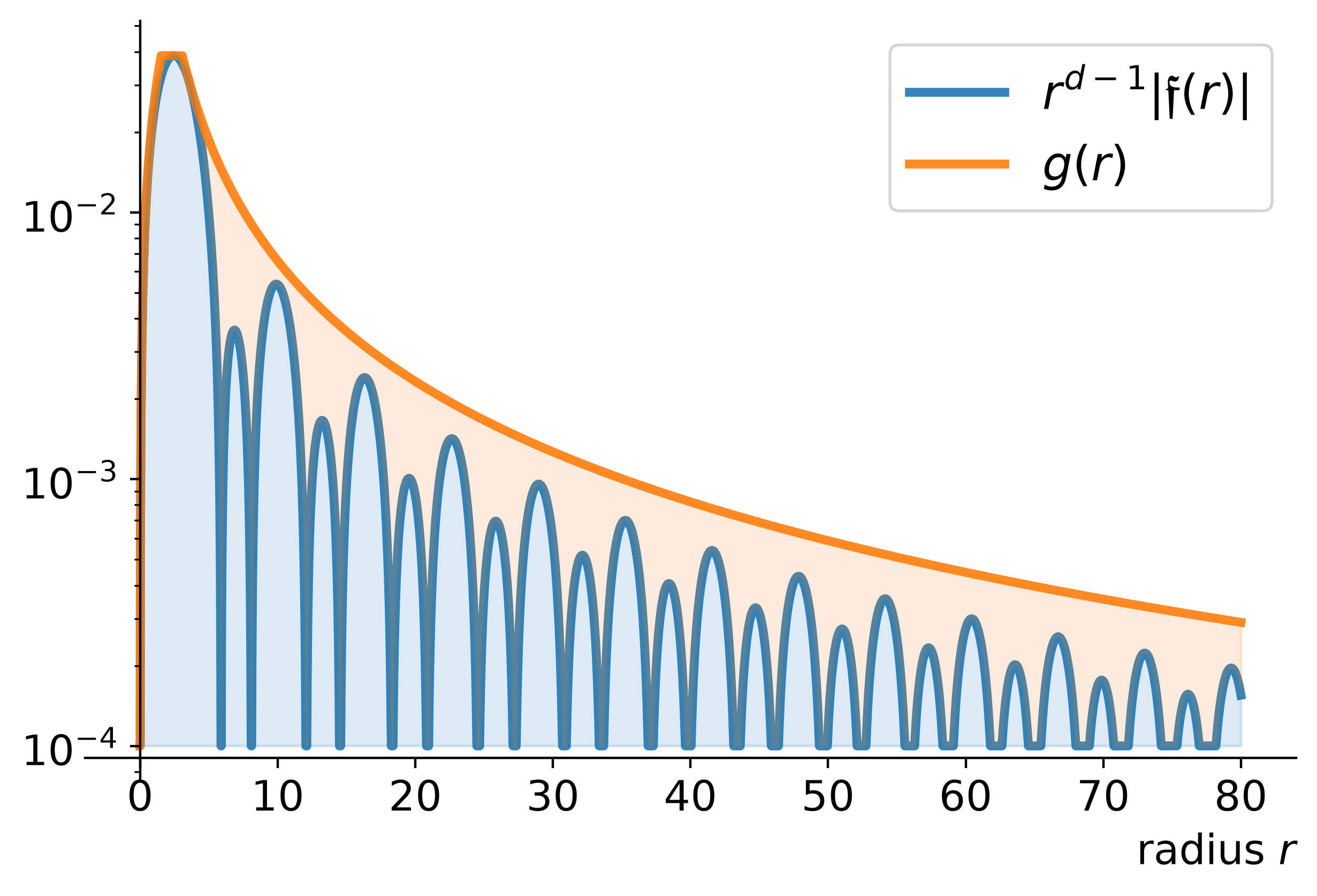}%
\end{minipage}

\caption{Proposal densities for the spectral radius densities proportional
to $r^{d-1}\left|\mathfrak{f}(r)\right|$ of the triangular kernel
(left: $d=1$, right: $d=2$).\label{fig:proposal_triangular}}
\end{figure}

\begin{figure}[H]
\begin{minipage}[t]{0.48\columnwidth}%
\includegraphics[width=0.38\paperwidth]{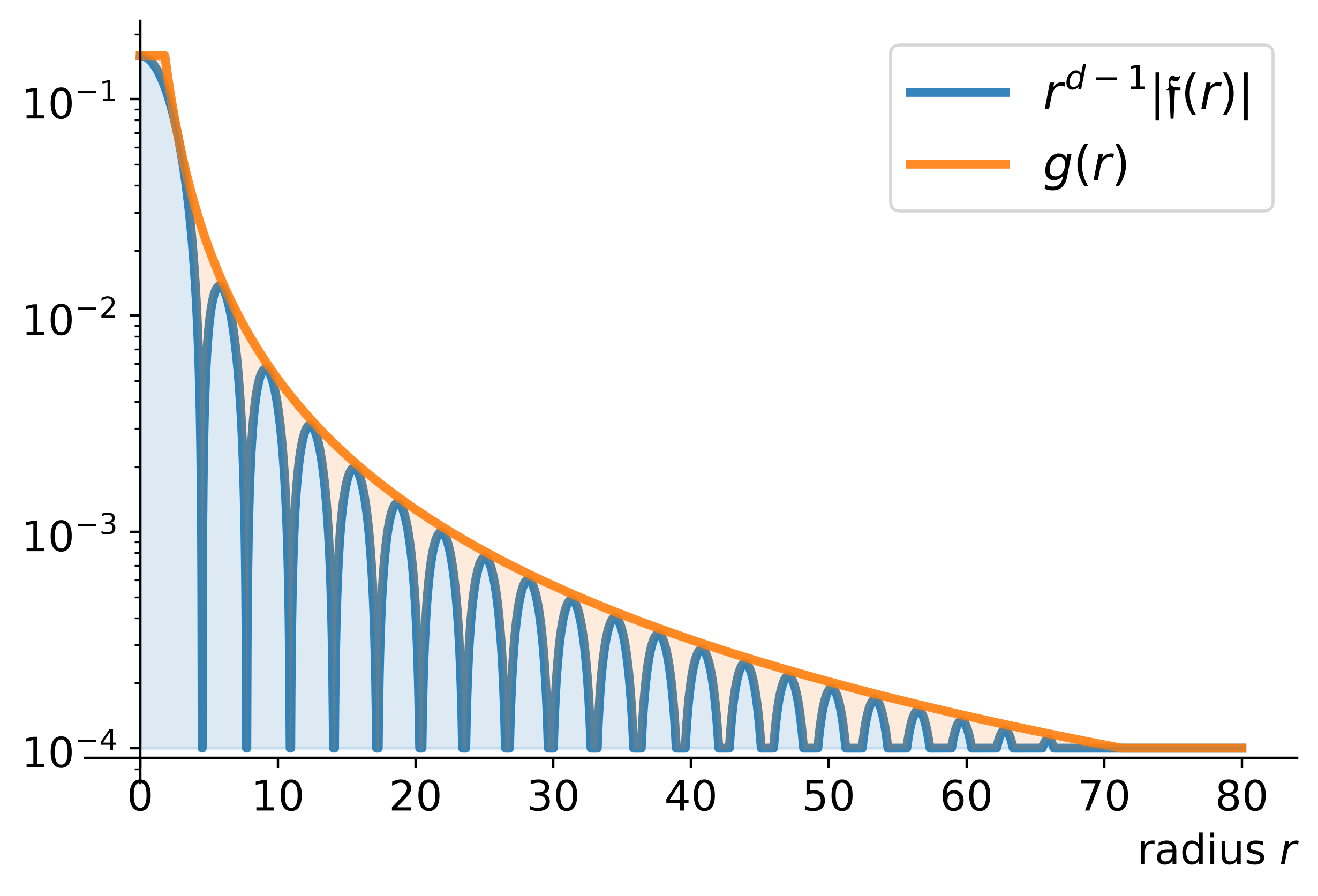}%
\end{minipage}\hfill{}%
\begin{minipage}[t]{0.48\columnwidth}%
\includegraphics[width=0.38\paperwidth]{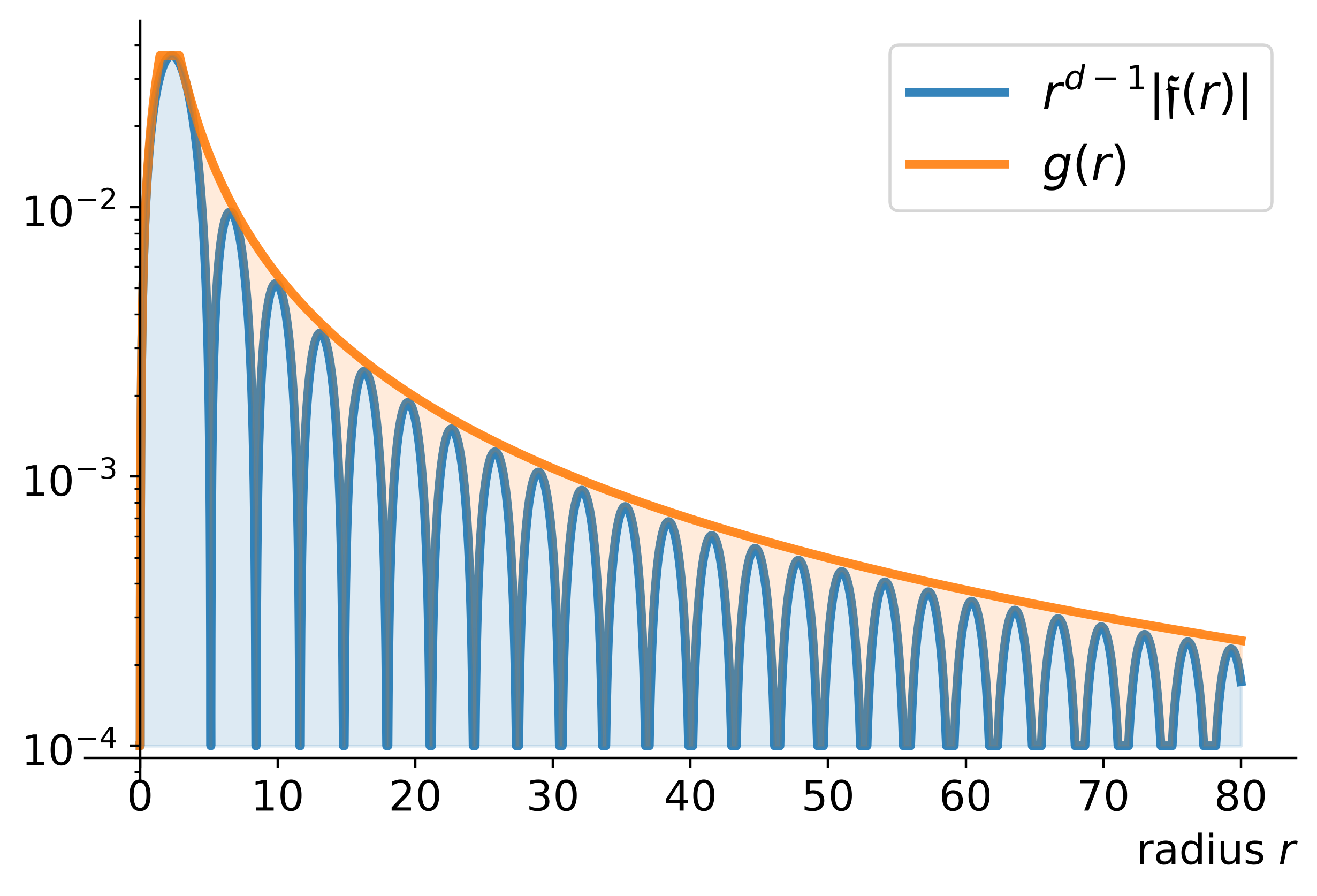}%
\end{minipage}

\caption{Proposal densities for the spectral radius densities proportional
to $r^{d-1}\left|\mathfrak{f}(r)\right|$ of the parabolic kernel
(left: $d=1$, right: $d=2$).\label{fig:proposal_parabolic}}
\end{figure}

\begin{figure}[H]
\begin{minipage}[t]{0.48\columnwidth}%
\includegraphics[width=0.38\paperwidth]{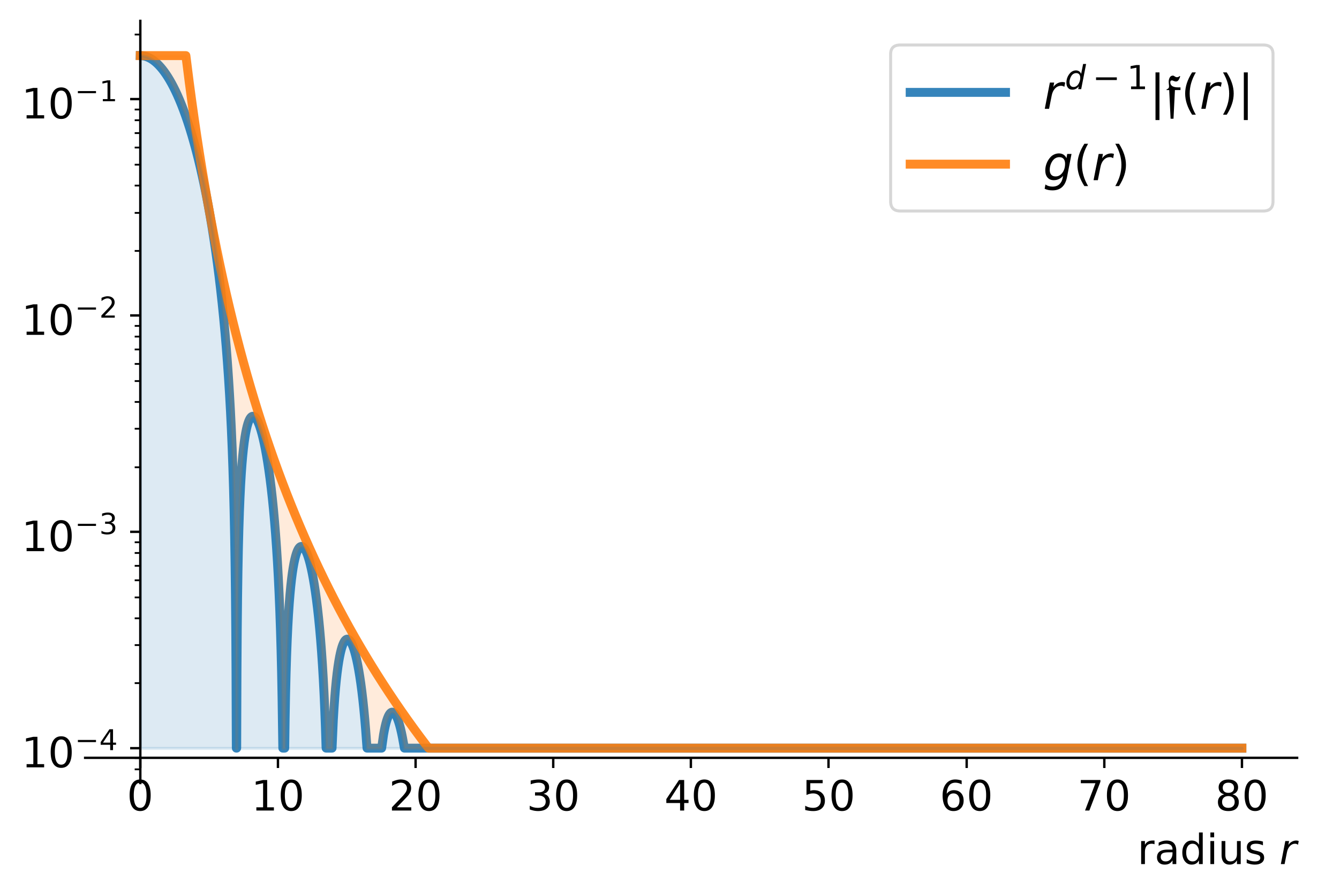}%
\end{minipage}\hfill{}%
\begin{minipage}[t]{0.48\columnwidth}%
\includegraphics[width=0.38\paperwidth]{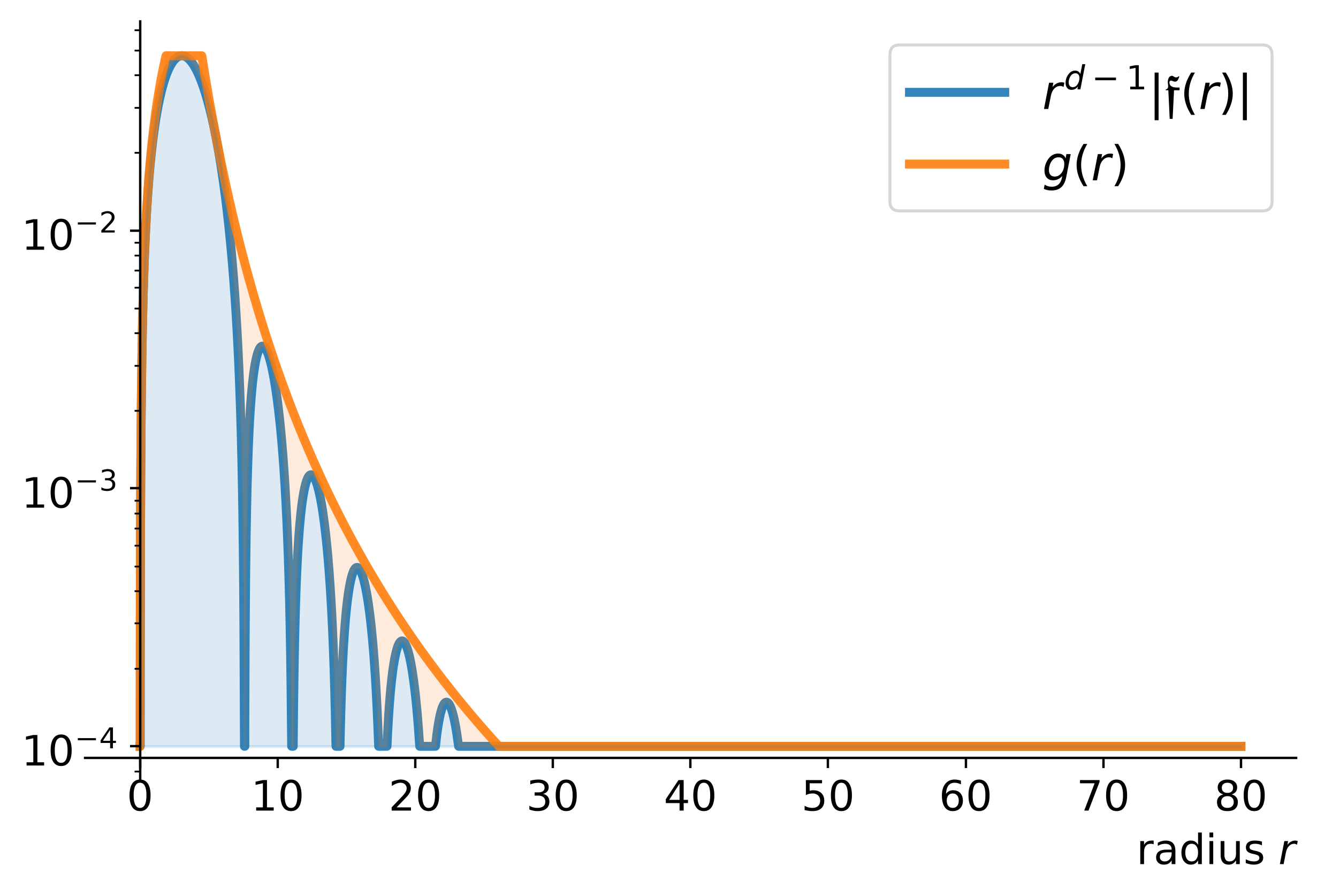}%
\end{minipage}

\begin{minipage}[t]{0.48\columnwidth}%
\includegraphics[width=0.38\paperwidth]{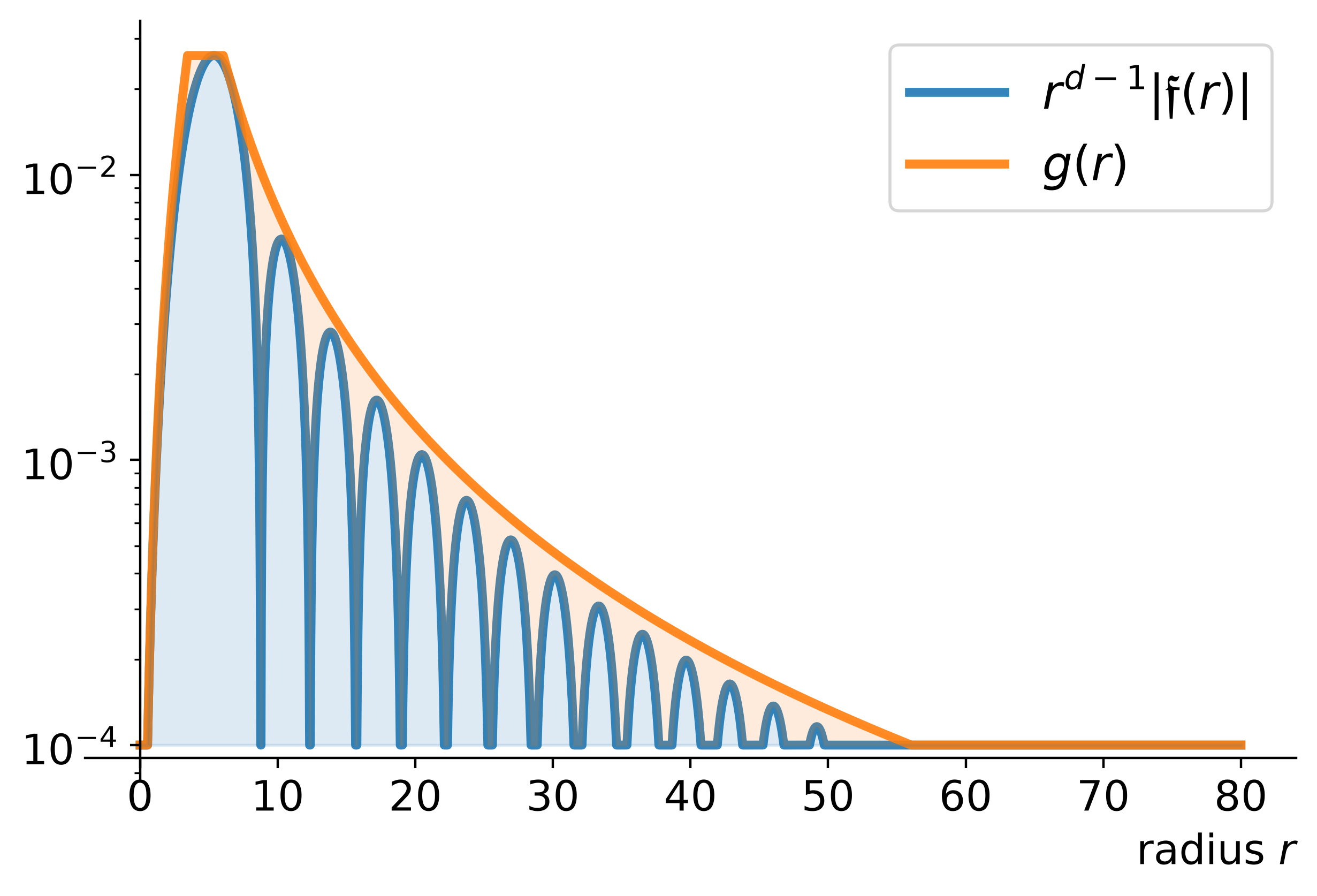}%
\end{minipage}\hfill{}%
\begin{minipage}[t]{0.48\columnwidth}%
\includegraphics[width=0.38\paperwidth]{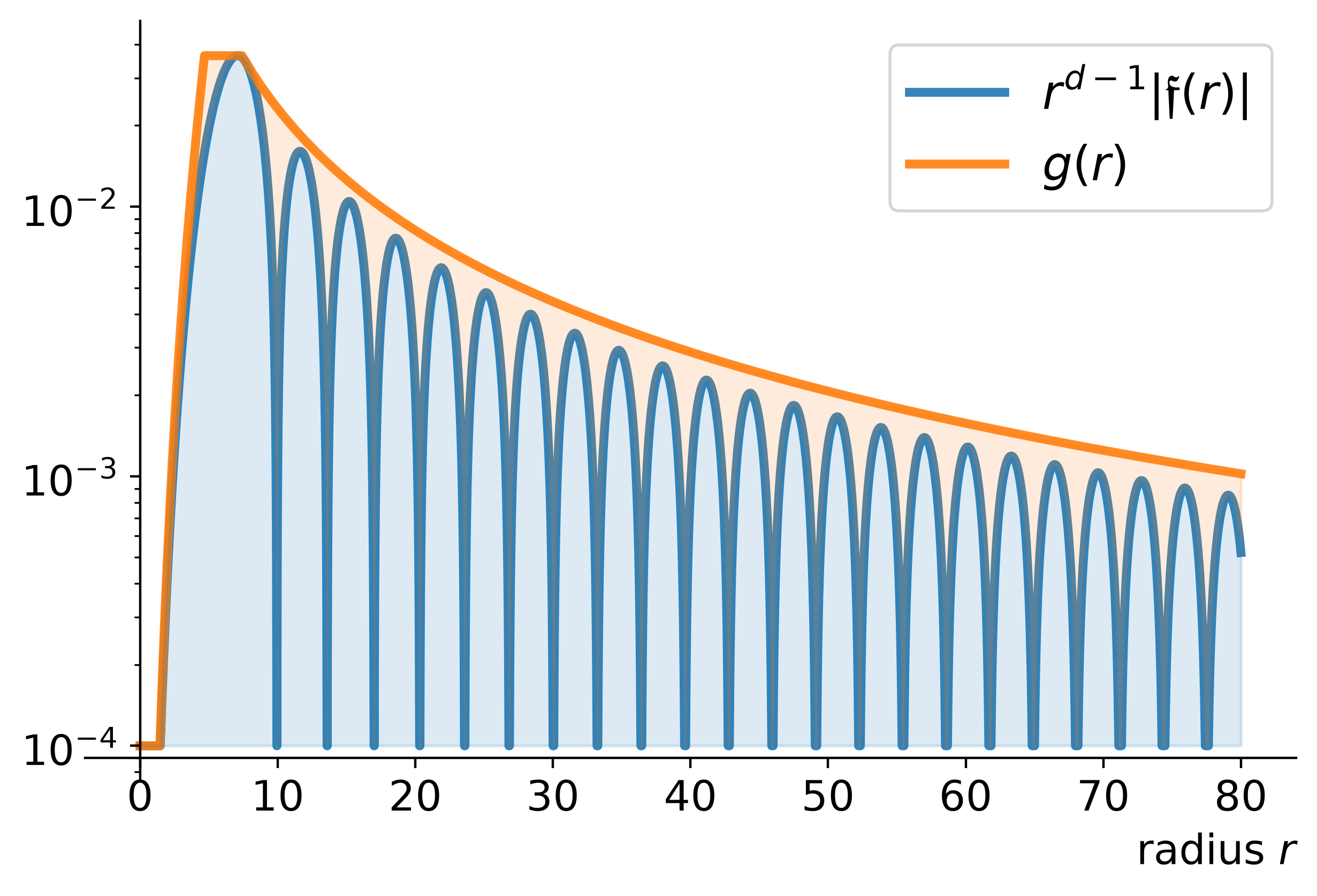}%
\end{minipage}

\caption{Proposal densities for the spectral radius densities proportional
to $r^{d-1}\left|\mathfrak{f}(r)\right|$ of the triweight kernel
(top left: $d=1$, top right: $d=2$, bottom left: $d=4$, bottom
right: $d=6$).\label{fig:proposal_triweight}}
\end{figure}

For comparison, Figure~\ref{fig:proposal_parabolic_fplus_fminus}
applies the same optimized envelope~\eqref{eq:proposal_envelope}
to the positive and negative parts $f_{+}$ and $f_{-}$ of the inverse Fourier
transform of the univariate parabolic kernel. This figure illustrates
why sampling from $f_{+}$ and $f_{-}$ separately by acceptance-rejection,
as done in \citep{liu2021fast,luo2021towards,he2024random} is less
efficient than sampling from $\left|f\right|$ (Figure~\ref{fig:proposal_parabolic}):
since $f_{+}$ and $f_{-}$ have non-overlapping support, the effective
rejection rate is much higher in the $f_{+}$/$f_{-}$ sampling approach
(Figure~\ref{fig:proposal_parabolic_fplus_fminus}, rejection is
guaranteed more than half of the time) than in the $\left|f\right|$
sampling approach (Figure~\ref{fig:proposal_parabolic}).

\begin{figure}[H]
\begin{minipage}[t]{0.48\columnwidth}%
\includegraphics[width=0.38\paperwidth]{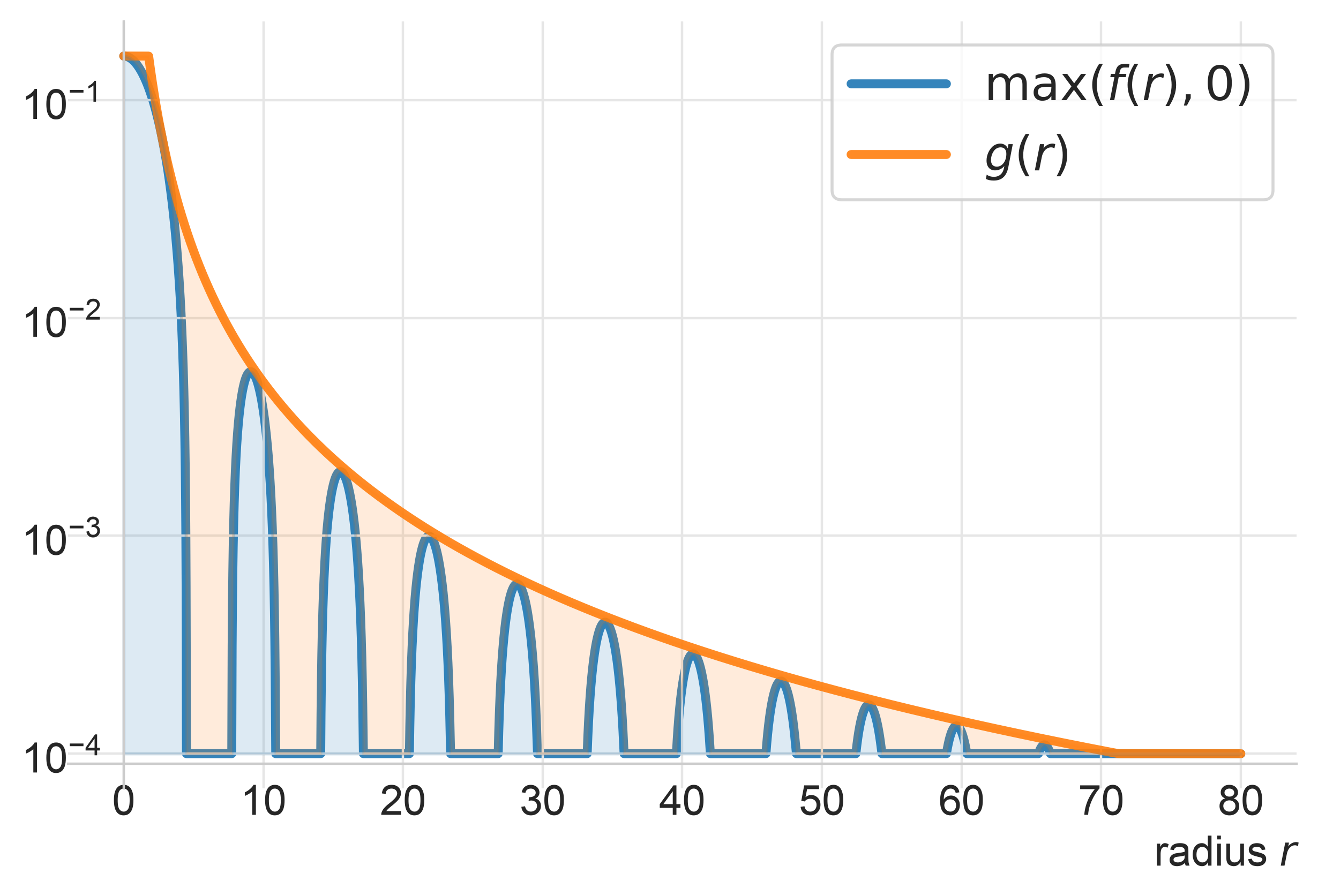}%
\end{minipage}\hfill{}%
\begin{minipage}[t]{0.48\columnwidth}%
\includegraphics[width=0.38\paperwidth]{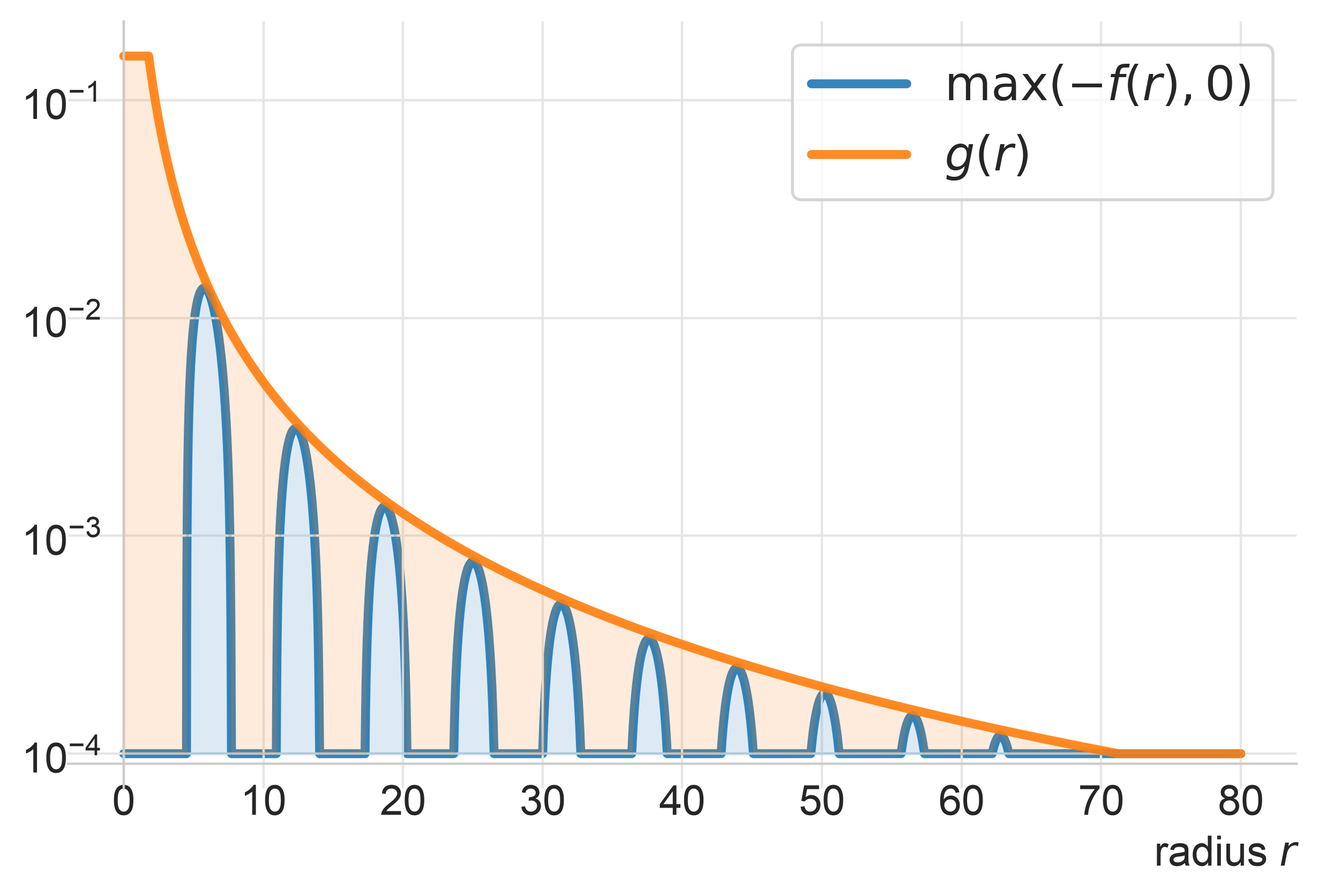}%
\end{minipage}

\caption{Proposal densities for the truncated spectral radius densities proportional
to the positive part $\max(f(r),0)$ (left) and the negative part
$\max(-f(r),0)$ (right) of the inverse Fourier transform of the univariate
parabolic kernel. To be compared to Figure~\ref{fig:proposal_parabolic}
(left).\label{fig:proposal_parabolic_fplus_fminus}}
\end{figure}

Finally, Figures \ref{fig:triangular_1d}-\ref{fig:triangular_2d}-\ref{fig:parabolic_1d}-\ref{fig:parabolic_2d}-\ref{fig:triweight_1d}-\ref{fig:triweight_2d}
show on some examples that the signed random Fourier features formula~\eqref{eq:SRFF_MC}
implemented with the spectral rejection sampling algorithms~\ref{algo:acceptance_rejection}-\ref{algo:proposal_sampling}
is indeed able to recover the shape of univariate
and bivariate indefinite kernels.

\begin{figure}[H]
\begin{centering}
\includegraphics[width=1\textwidth]{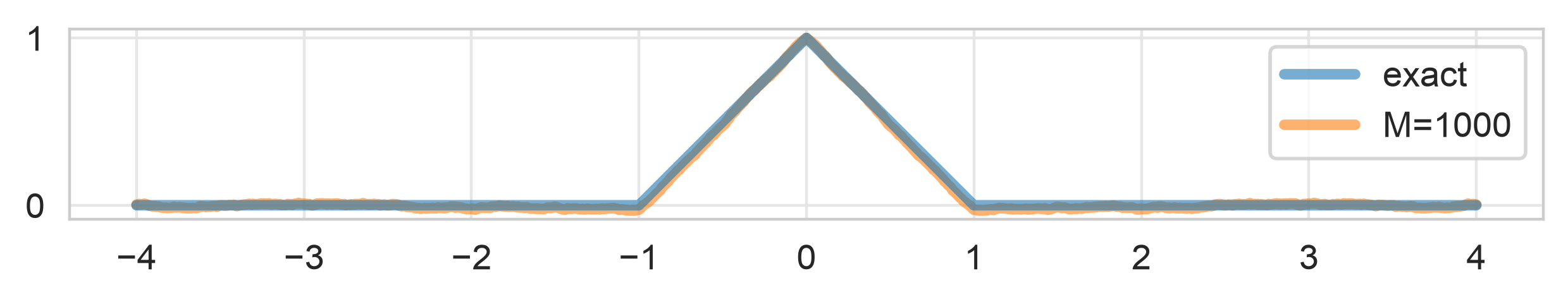}
\par\end{centering}
\vspace{-2mm}
\caption{Univariate triangular kernel and its random Fourier features approximation
\eqref{eq:random_fourier_features} using $M=1000$ random projections.\label{fig:triangular_1d}}
\end{figure}

\begin{figure}[H]
\begin{centering}
\includegraphics[width=1\textwidth, trim=0 0 0 16mm, clip]{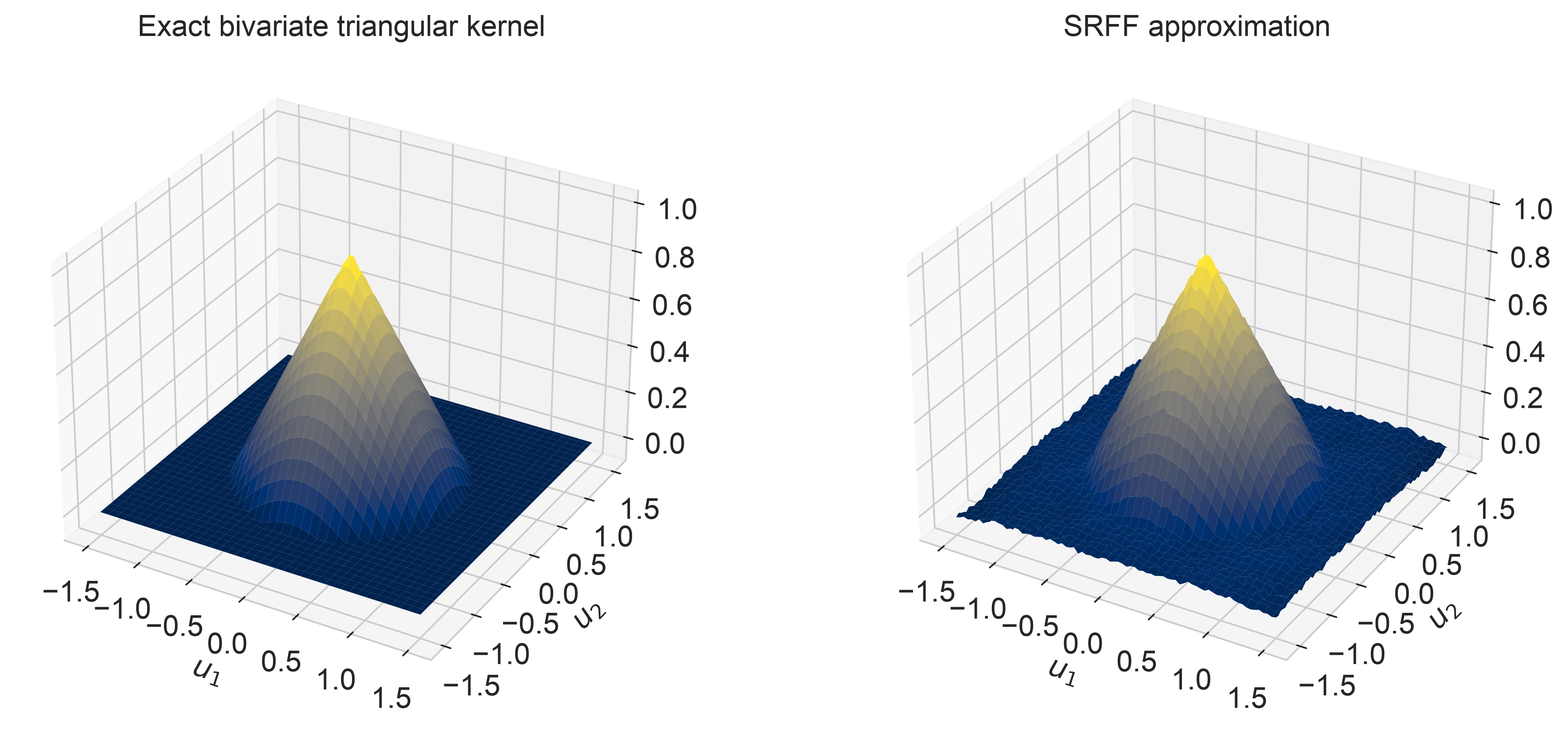}
\par\end{centering}
\caption{Bivariate triangular kernel (left) and its signed random Fourier features
approximation \eqref{eq:SRFF_MC} (right) using $M=4000$ random projections.\label{fig:triangular_2d}}
\end{figure}

\begin{figure}[H]
\begin{centering}
\includegraphics[width=1\textwidth]{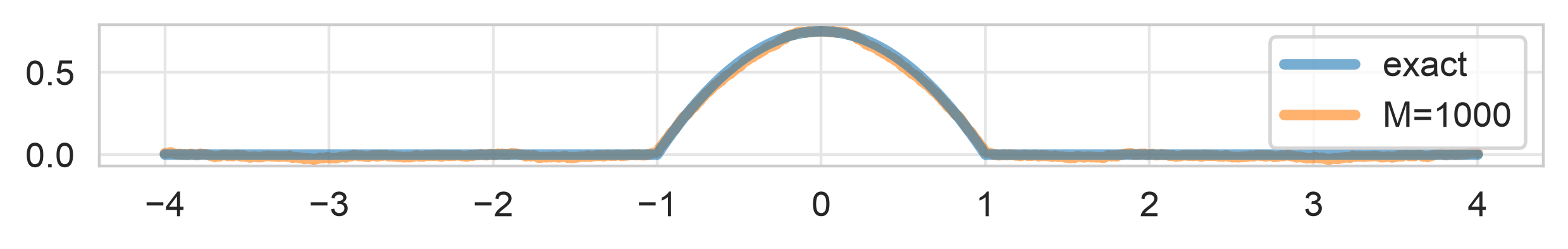}
\par\end{centering}
\vspace{-3mm}
\caption{Univariate parabolic kernel and its signed random Fourier features
approximation \eqref{eq:SRFF_MC} using $M=1000$ random projections.\label{fig:parabolic_1d}}
\end{figure}
\vspace{-2mm}
\begin{figure}[H]
\begin{centering}
\includegraphics[width=1\textwidth, trim=0 0 0 16mm, clip]{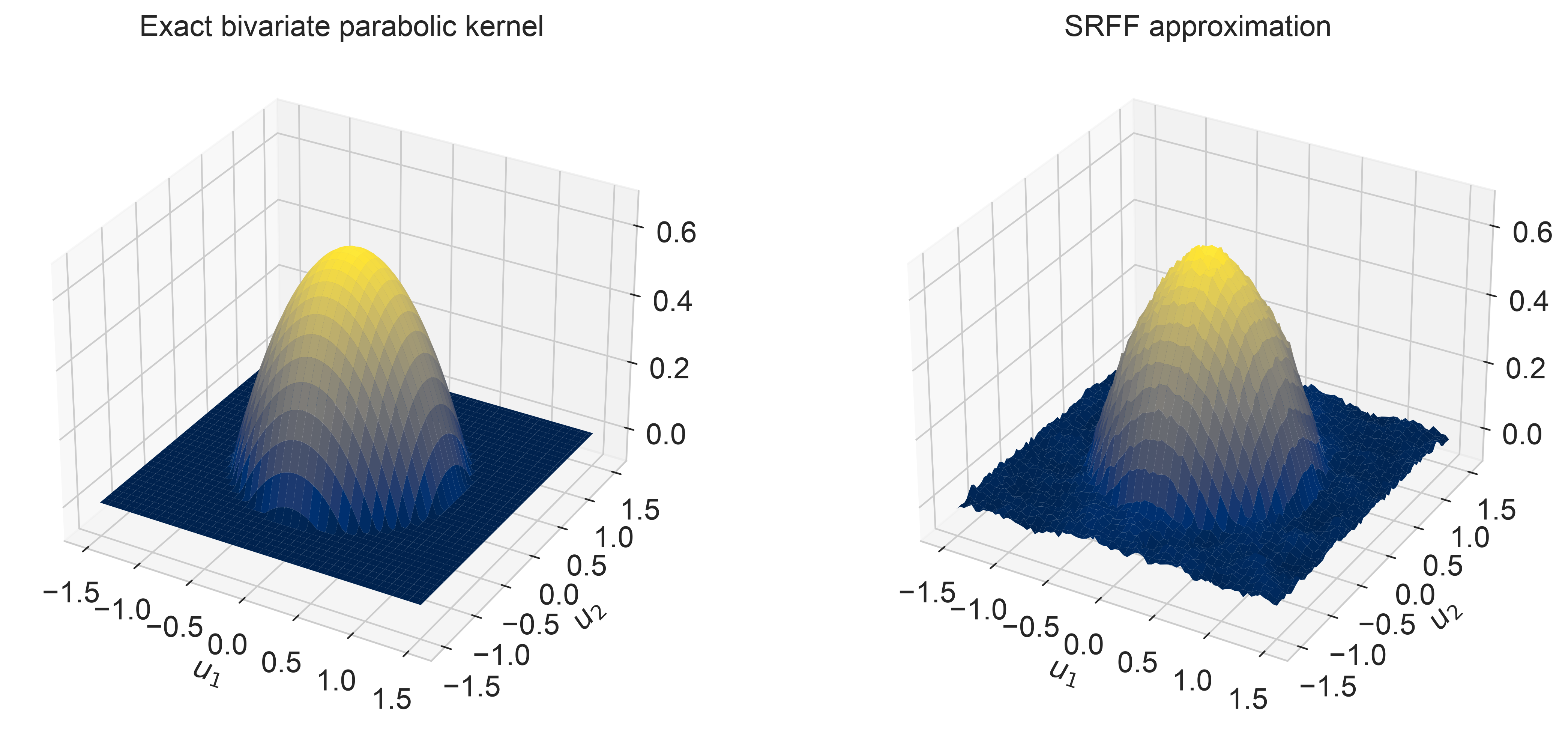}
\par\end{centering}
\vspace{-3mm}
\caption{Bivariate parabolic kernel (left) and its signed random Fourier features
approximation \eqref{eq:SRFF_MC} (right) using $M=4000$ random projections.\label{fig:parabolic_2d}}
\end{figure}
\vspace{-2mm}
\begin{figure}[H]
\begin{centering}
\includegraphics[width=1\textwidth]{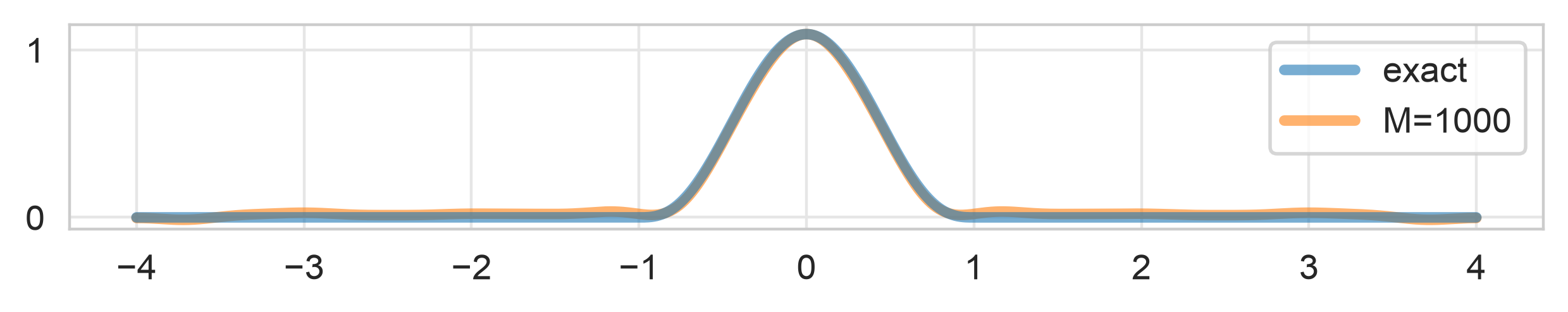}
\par\end{centering}
\vspace{-3mm}
\caption{Univariate triweight kernel and its signed random Fourier features
approximation \eqref{eq:SRFF_MC} using $M=1000$ random projections.\label{fig:triweight_1d}}
\end{figure}
\vspace{-2mm}
\begin{figure}[H]
\begin{centering}
\includegraphics[width=1\textwidth, trim=0 0 0 16mm, clip]{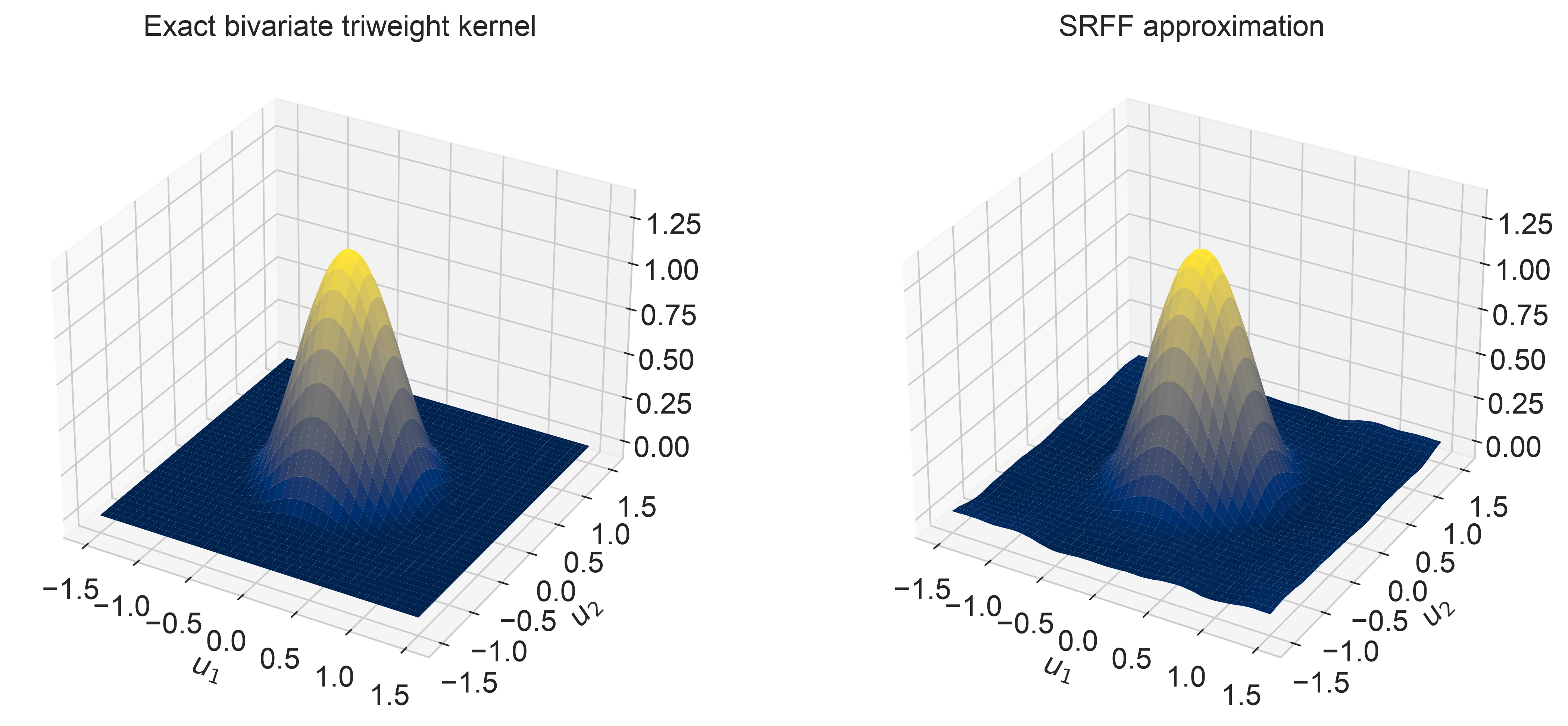}
\par\end{centering}
\vspace{-2mm}
\caption{Bivariate triweight kernel (left) and its signed random Fourier features
approximation \eqref{eq:SRFF_MC} (right) using $M=4000$ random projections.\label{fig:triweight_2d}}
\end{figure}

\section{Numerical experiments\label{sec:numerical}}

In this section, we conduct numerical tests to evaluate the performance
and accuracy of the fast kernel density estimation Algorithm~\ref{algo:fast_kernel_mvm}.
The goal is to compare the results obtained using SRFF with those
produced by direct summation, and to evaluate both the computational
efficiency and the precision of this approximation method. 

\subsection*{Computational setup}

The accuracy experiments are performed on a laptop equipped with a
13th Gen Intel(R) Core(TM) i9-13900H processor running at 2.60 GHz
and 32.0 GB of RAM (31.7 GB usable). The runtime experiments are performed
on a Linux compute node with 64 CPU
cores (128 threads), a base clock speed of 2.9 GHz, a maximum clock
speed of 4.0 GHz, 2048 GB of RAM, and 9.7 TB of scratch disk space.
The experiments use 16 CPU threads. Algorithm~\ref{algo:fast_kernel_mvm}
is implemented in C++. Runtime is measured using high-resolution clock
functions.

\subsection*{Error metric}

To measure the approximation error of the SRFF
algorithm, we compute the mean absolute error (MAE) between the kernel
MVM~\eqref{eq:kernel_mvm} and its SRFF
approximation computed by Algorithm~\ref{algo:fast_kernel_mvm},
averaged over the whole dataset $\mathcal{D}=\left\{ (\boldsymbol{x}_{n},y_{n})\in\mathbb{R}^{d}\times\mathbb{R}\ensuremath{,}n=1,\ldots,N\right\} $:
\[
\mathrm{MAE}((\bm{\eta}_{1},\ldots,\bm{\eta}_{M})):=\frac{1}{N}\sum^{N}_{n=1}\left|\left(\sum^{N}_{n'=1}y_{n'}K(\boldsymbol{x}_{n}-\boldsymbol{x}_{n'})\right)-\left(\sum^{N}_{n'=1}y_{n'}K_{M}(\boldsymbol{x}_{n}-\boldsymbol{x}_{n'})\right)\right|,
\]
which depends on the $M$ simulations $\bm{\eta}_{1},\ldots,\bm{\eta}_{M}$
of the random projection vector $\bm{\eta}$. Since this $\mathrm{MAE}$
is random, in practice we resimulate these $M$ random projections
$L$ times independently, and report the sample average
\[
\mathrm{MAE}:=\frac{1}{L}\sum^{L}_{\ell=1}\mathrm{MAE}((\bm{\eta}_{1,(\ell)},\ldots,\bm{\eta}_{M,(\ell)}))
\]
where $\bm{\eta}_{(\ell)}=(\bm{\eta}_{1,(\ell)},\ldots,\bm{\eta}_{M,(\ell)})$,
$\ell=1,\ldots,L$, are $L$ independent copies of the set of $M$
independent simulations $(\bm{\eta}_{1},\ldots,\bm{\eta}_{M})$.

\subsection*{Accuracy evaluation}

We evaluate the accuracy of SRFF using the MAE defined above. We fix $N=1{,}000{,}000$ and
consider $d=1$ and $d=2$ as examples. For each dimension, we test
$28$ values of $M$, ranging from $17$ to $196{,}609$. For each
pair of $d$ and $M$, we report the mean MAE over $L=5$ distinct
random trials. Within each trial, we generate one master sequence
containing $196{,}609$ random frequencies. For a given $M$, the
estimator uses the first $M$ frequencies from this sequence. The
estimates are therefore nested across different values of $M$ within
the same trial, while different trials use independent master sequences.

\begin{figure}[H]
\begin{minipage}[t]{0.48\columnwidth}%
\includegraphics[width=0.38\paperwidth]{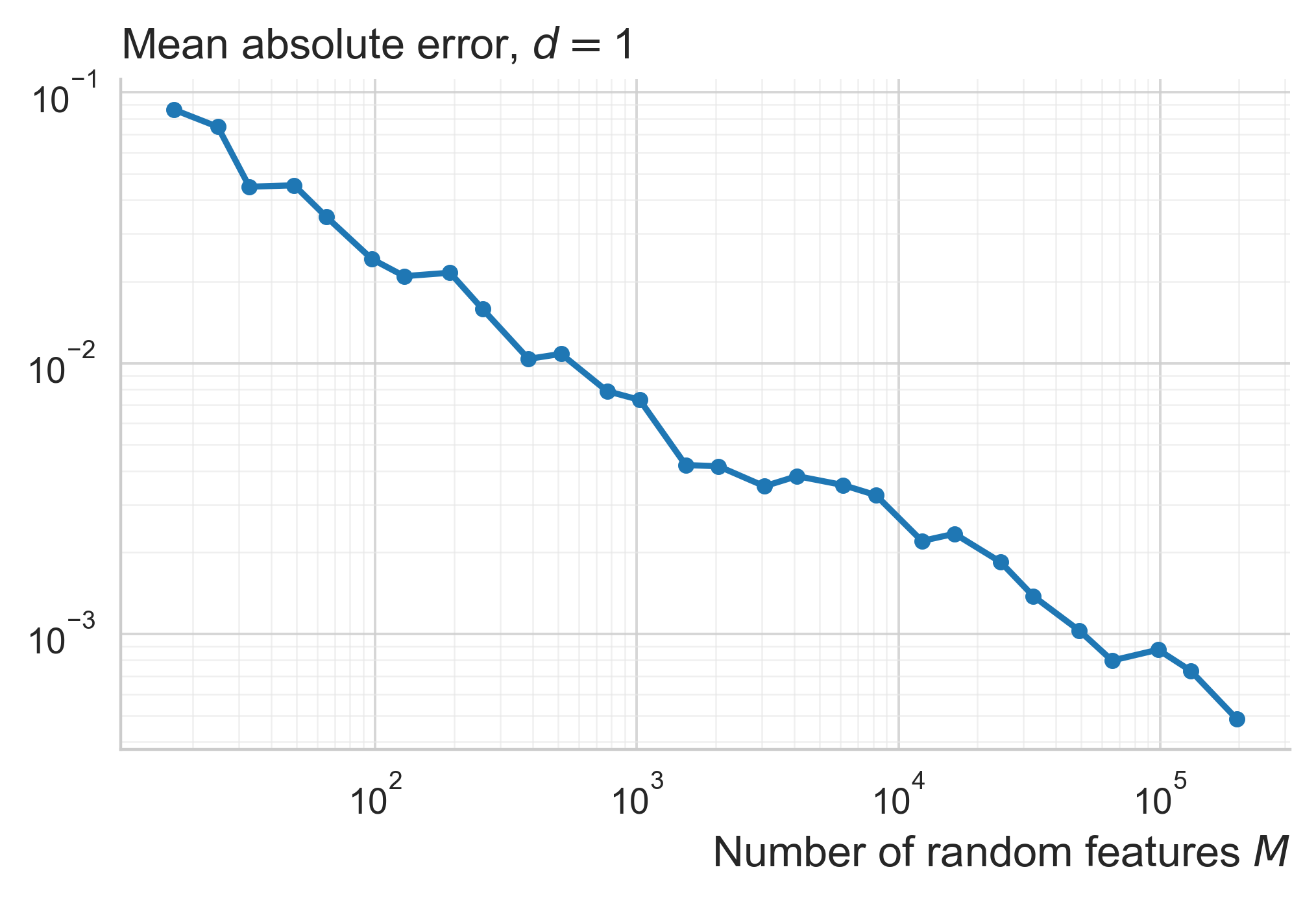}%
\end{minipage}\hfill{}%
\begin{minipage}[t]{0.48\columnwidth}%
\includegraphics[width=0.38\paperwidth]{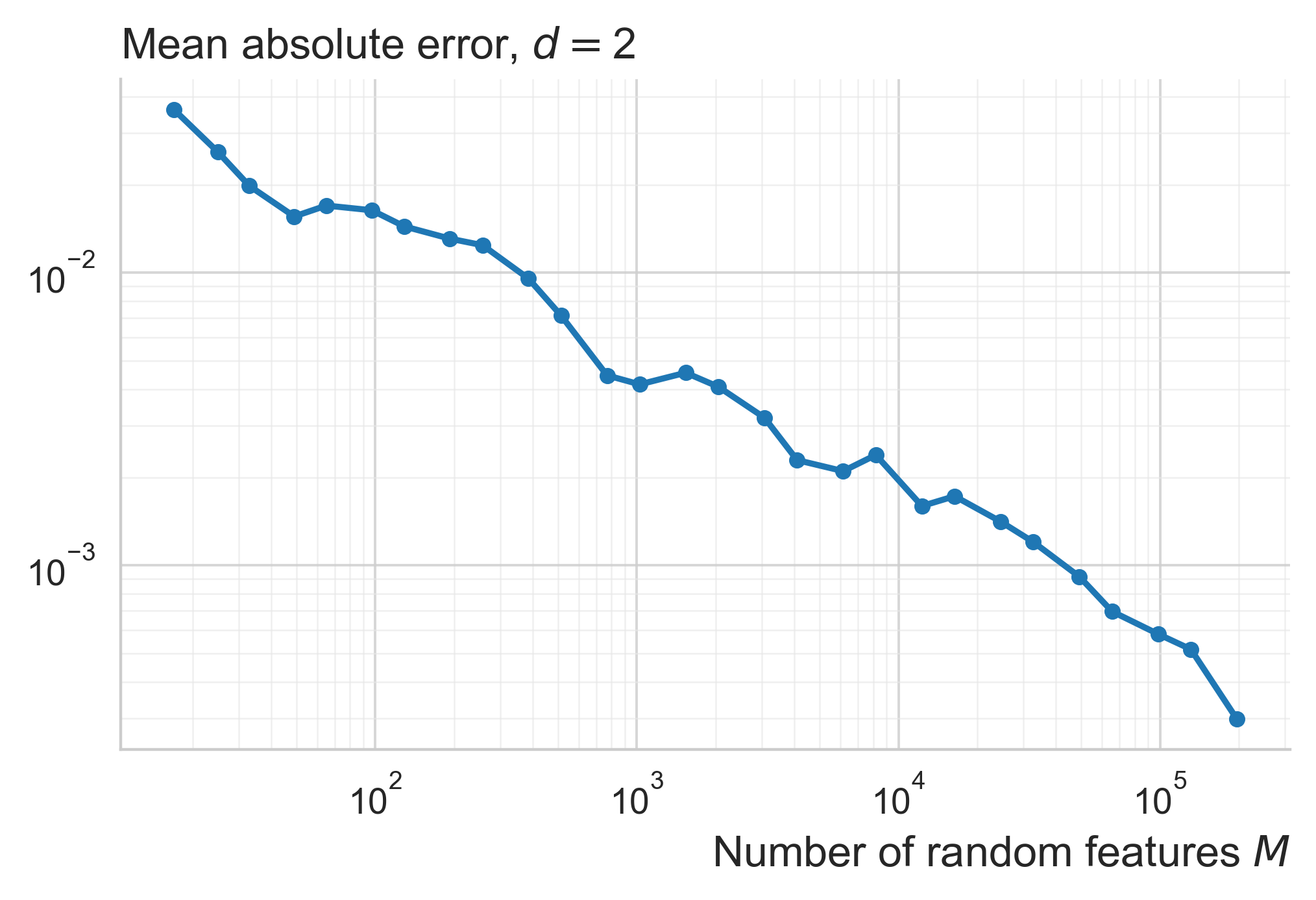}%
\end{minipage}

\caption{Mean absolute error of SRFF with respect to $M$ (left: $d=1$, right:
$d=2$).\label{fig:accuracy}}
\end{figure}

Figure~\ref{fig:accuracy} shows, as expected, an overall decrease
in MAE as $M$ increases. With $L=5$, small local fluctuations are
still visible, but the log--log trends have slopes close to $-1/2$,
as expected from the Monte Carlo approximation formula \eqref{eq:SRFF_MC}.
When $M=196{,}609$, the average MAE is approximately $4.85\times10^{-4}$
for $d=1$ and $3.0\times10^{-4}$ for $d=2$.

\subsection*{Speed evaluation}

To assess computational efficiency, we compare the runtime of direct
KDE, with complexity $O(N^{2})$, with that using SRFF decomposition,
with complexity $O(NM)$. In this numerical experiment using the parabolic
kernel, we set $N=1{,}000{,}000$ and consider $d=1$ and $d=2$.
The reported runtimes are averaged over seven repeated experiments.

\begin{figure}[H]
\begin{minipage}[t]{0.48\columnwidth}%
\includegraphics[width=0.38\paperwidth]{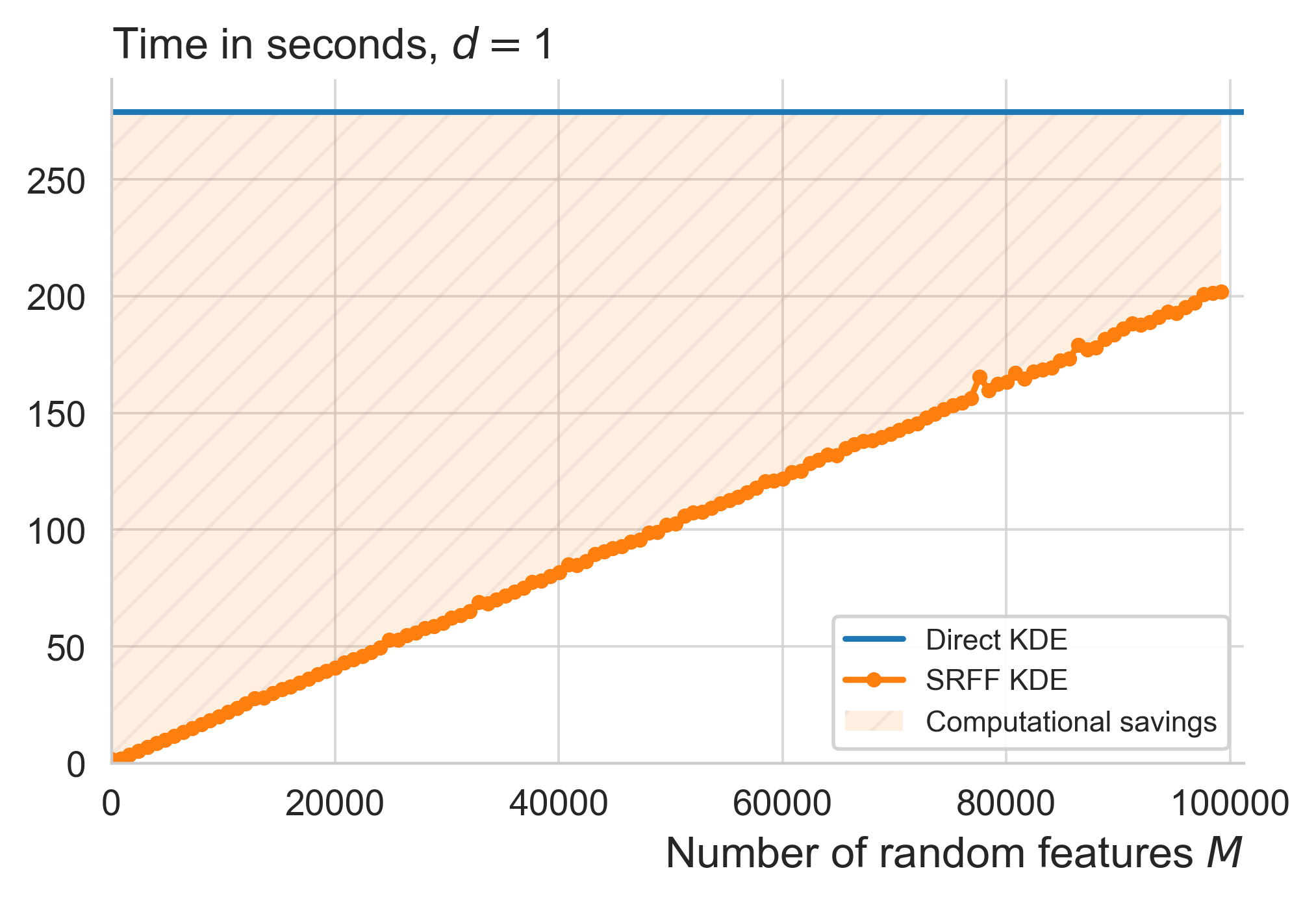}%
\end{minipage}\hfill{}%
\begin{minipage}[t]{0.48\columnwidth}%
\includegraphics[width=0.38\paperwidth]{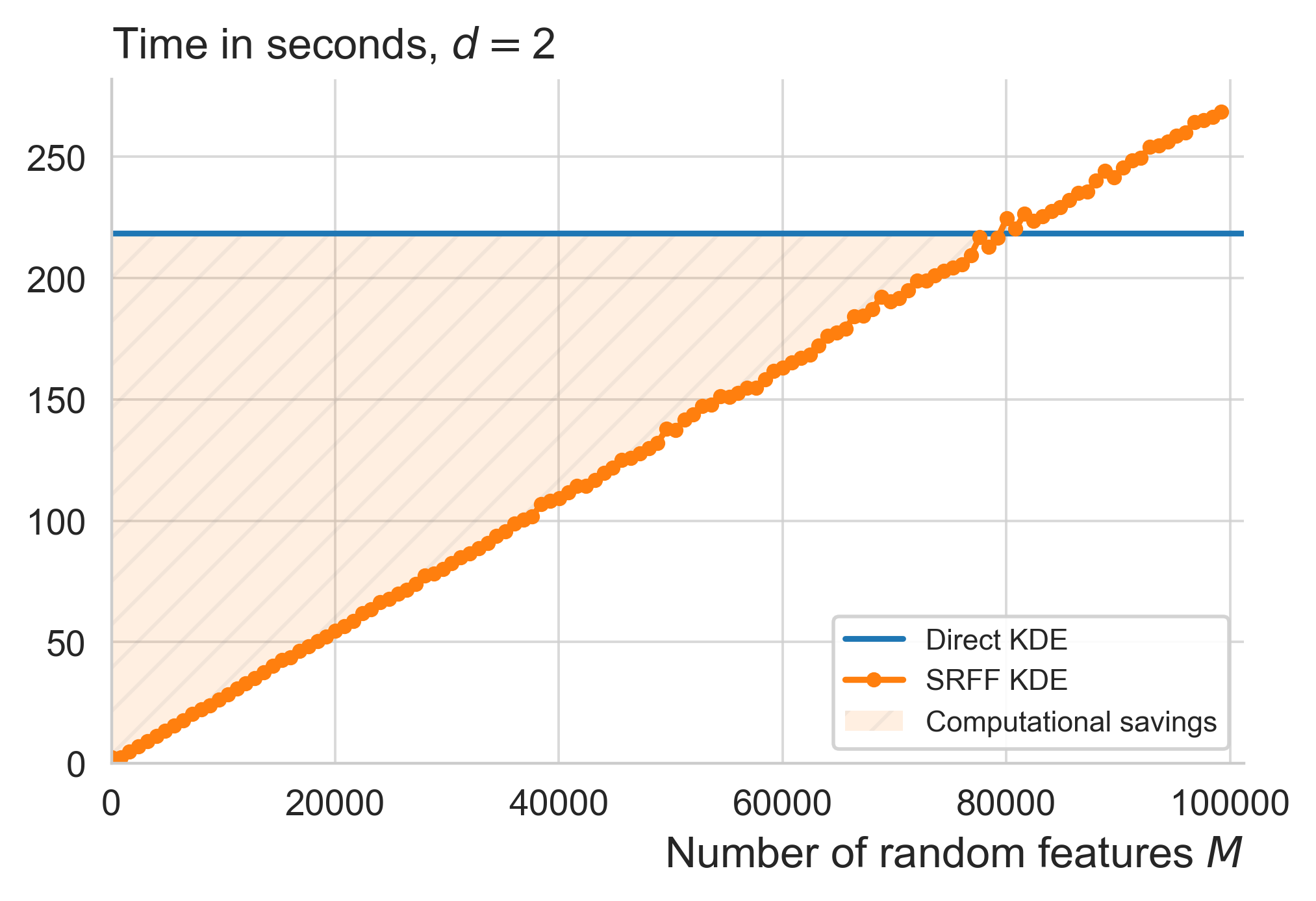}%
\end{minipage}

\caption{Runtime comparison between direct KDE and SRFF KDE with respect to
the number of random features $M$ (left: $d=1$, right: $d=2$).\label{fig:speed}}
\end{figure}

Figure~\ref{fig:speed} shows that the SRFF runtime increases approximately
linearly with $M$ in both dimensions, as expected. For any reasonable
value of $M$ (usually set to a few hundreds or a few thousands in
the literature), SRFF provides huge computational savings compared
to direct KDE. For example, when $M=1000$, KDE is computed on $N=1{,}000{,}000$ points in about 2 seconds using SRFF, which is more than a hundred times faster than direct KDE evaluation.
By combining the results from Figures~\ref{fig:accuracy}
and~\ref{fig:speed}, one can adjust the value of $M$ to achieve
any desired trade-off between approximation accuracy and computational
cost.

\section{Conclusion\label{sec:conclusion}}

This article introduced the signed random Fourier features technique
\eqref{eq:SRFF_MC} (SRFF), which is an extension of the classical
random Fourier features technique \citep{rahimi2007random} (RFF)
to kernel function that are not positive definite, such as most compact
kernel functions encountered in kernel density estimation (Table~\ref{tab:kuttner_golubov}).
We provided a detailed algorithm to simulate the signed spectral distribution
of such kernels (Algorithm~\ref{algo:proposal_sampling}) in the
case of the general class of compact multivariate Kuttner-Golubov
kernels \eqref{eq:kuttner_golubov_kernel}, for which we provided an
analytical formula \eqref{eq:fourier_kuttner_golubov} for its Fourier
transform. We showed that sampling from the absolute value $\left|f\right|/\left\Vert f\right\Vert _{1}$
of the inverse Fourier transform $f$ of the indefinite kernel $K$ under
consideration by acceptance-rejection is more efficient than sampling
from its positive part $f_{+}/\left\Vert f_{+}\right\Vert _{1}$ and
negative part $f_{-}/\left\Vert f_{-}\right\Vert _{1}$ independently,
as done so far in the literature. Moreover, our approach does not
require to compute the constants $\left\Vert f_{+}\right\Vert _{1}$,
$\left\Vert f_{-}\right\Vert _{1}$, or $\left\Vert f\right\Vert _{1}$
analytically. Then, we constructed an algorithm to perform fast kernel
matrix-vector multiplication using SRFF (Algorithm~\ref{algo:fast_kernel_mvm}),
which includes kernel density estimation as a particular case. Our
numerical tests confirm the speed and accuracy of SRFF for kernel
density estimation. More generally, the SRFF methodology can be applied
to any kernel-based estimator built upon indefinite kernels, and the
underpinning signed Monte Carlo methodology \eqref{eq:EgX_SMC}
can be applied to any sampling estimator built upon signed densities.

\section*{Acknowledgements}
Wen Chen acknowledges the support of the BNBU Start-up Research Fund UICR0700082-24. Nicolas Langren\'e and Wen Chen acknowledge the support of the Guangdong Provincial/Zhuhai Key Laboratory of IRADS (2022B1212010006).

\bibliographystyle{apalike}
\bibliography{biblio}

\appendix

\section{Fourier transforms of multivariate compact kernels\label{sec:fourier_transforms}}

Let $K(\boldsymbol{u})=k(\left\Vert \boldsymbol{u}\right\Vert )$,
$\boldsymbol{u}\in\mathbb{R}^{d}$, be a continuous, isotropic (a.k.a.
radial) kernel in $\mathbb{R}^{d}$. Then its Fourier transform $f$
\eqref{eq:f_wrt_K} is also radial, and given by \citep[Theorem~B.1]{fasshauer2007meshfree}:
\begin{equation}
f(\boldsymbol{x})=\frac{1}{(2\pi)^{\frac{d}{2}}\left\Vert \boldsymbol{x}\right\Vert ^{\frac{d}{2}-1}}\int^{\infty}_{0}k(t)t^{\frac{d}{2}}J_{\frac{d}{2}-1}(\left\Vert \boldsymbol{x}\right\Vert t)dt\ ,\ \boldsymbol{x}\in\mathbb{R}^{d},\label{eq:kernel_fourier_transform}
\end{equation}
where $J_{\nu}$ is the Bessel function of the first kind with order
$\nu$. Bochner's theorem tells us that equation~\eqref{eq:kernel_fourier_transform}
is a density, known as spectral density, if and only if the kernel
$K$ is positive definite. In this Appendix, we are going to compute
equation~\eqref{eq:kernel_fourier_transform} for several kernels
of interest in density estimation.

\subsection{Kuttner-Golubov kernels\label{subsec:kuttner_golubov_kernels}}

Recall the definition of the Kuttner-Golubov kernel~\eqref{eq:kuttner_golubov_kernel}:
\[
K(\boldsymbol{u})=k(\left\Vert \boldsymbol{u}\right\Vert )=(1-\left\Vert \boldsymbol{u}\right\Vert ^{\alpha})^{\beta}\mathbbm{1}_{\{\left\Vert \boldsymbol{u}\right\Vert \leq1\}}\ ,\ \boldsymbol{u}\in\mathbb{R}^{d},
\]
where $\alpha>0$ and $\beta>0$. First, Lemma~\ref{lem:kernel_integral}
computes the integral of this kernel, which is needed whenever a kernel
estimator requires the kernel function to be scaled as a density.

\begin{lem}
\label{lem:kernel_integral}The integral of the Kuttner-Golubov kernel
function $K$ defined in equation~\eqref{eq:kuttner_golubov_kernel}
is given explicitly by
\begin{equation}
\int_{\mathbb{R}^{d}}K(\boldsymbol{u})d\boldsymbol{u}=\int_{\left\Vert \boldsymbol{u}\right\Vert \leq1}(1-\left\Vert \boldsymbol{u}\right\Vert ^{\alpha})^{\beta}d\boldsymbol{u}=\frac{2}{\alpha}\frac{\pi^{d/2}}{\Gamma\!\left(\frac{d}{2}\right)}\frac{\Gamma\!\left(\frac{d}{\alpha}\right)\Gamma(\beta+1)}{\Gamma\!\left(\frac{d}{\alpha}+\beta+1\right)}\ .\label{eq:kernel_integral}
\end{equation}
\end{lem}

\begin{proof}
Using a polar change of variable \citep[Corollary 2.51 and Proposition 2.54 page 79]{folland1999real},
\begin{align}
\int_{\mathbb{R}^{d}}K(\boldsymbol{u})d\boldsymbol{u} & =\int_{\mathbb{R}^{d}}(1-\left\Vert \boldsymbol{u}\right\Vert ^{\alpha})^{\beta}\mathbbm{1}_{\{\left\Vert \boldsymbol{u}\right\Vert \leq1\}}d\boldsymbol{u}\nonumber \\
 & =\frac{2\pi^{d/2}}{\Gamma\!\left(\frac{d}{2}\right)}\int^{\infty}_{0}r^{d-1}(1-r^{\alpha})^{\beta}\mathbbm{1}_{\{r\leq1\}}dr\ .\label{eq:kernel_integral_proof_1}
\end{align}
Then, 
\begin{equation}
\int^{1}_{0}r^{d-1}(1-r^{\alpha})^{\beta}dr=\frac{1}{\alpha}\int^{1}_{0}r^{\frac{d}{\alpha}-1}(1-r)^{\beta}dr=\frac{\Gamma\!\left(\frac{d}{\alpha}\right)\Gamma(\beta+1)}{\Gamma\!\left(\frac{d}{\alpha}+\beta+1\right)}\ .\label{eq:kernel_integral_proof_2}
\end{equation}
Plugging equation~\eqref{eq:kernel_integral_proof_2} into equation~\eqref{eq:kernel_integral_proof_1}
proves equation~\eqref{eq:kernel_integral}.
\end{proof}

The normalizing constant~\eqref{eq:kernel_integral} generalizes
the one obtained in \citep{duong2015beta} in the case $\alpha=2$.
The next proposition provides an explicit formula for the Fourier
transform of the kernel~\eqref{eq:kuttner_golubov_kernel}. Remark
that the Fourier transform formulas \eqref{eq:fourier_kuttner_golubov}-\eqref{eq:fourier_symmetric_beta}-\eqref{eq:fourier_askey}
provided in this Appendix can all be deduced from the formulas reported
in \citep[page~1204]{zastavnyi2006buhmann} with the change of notation
$(\delta,\mu,\nu,\alpha,t)=(\alpha,\beta+1,(d-1)/2,d-1,\left\Vert x\right\Vert )$
and by simplifying the resulting formulas. We nevertheless provide
our more detailed proofs for reference and comprehensiveness, and
because the proof can easily be adapted to other kernel functions,
as long as the integral $\int^{\infty}_{0}t^{\lambda}k(t)dt$, $\lambda\geq0$,
can be computed analytically (as done for example in \citep[subsection~14.3]{davies2002integral}
with $k(t)=(1-t^{2})^{\beta}$, covered here as a corollary in Subsection~\ref{subsec:symmetric_beta_kernels}).
\begin{prop}
\label{prop:fourier_kuttner_golubov}The Fourier transform of the
kernel $K(\boldsymbol{u})=(1-\left\Vert \boldsymbol{u}\right\Vert ^{\alpha})^{\beta}\mathbbm{1}_{\{\left\Vert \boldsymbol{u}\right\Vert \leq1\}}$,
$\boldsymbol{u}\in\mathbb{R}^{d}$, is given by
\begin{equation}
f(\boldsymbol{x})=\frac{\Gamma(\beta+1)}{\alpha2^{\frac{d}{2}-1}(2\pi)^{\frac{d}{2}}}\Psi_{1,2}\left[-\frac{\left\Vert \boldsymbol{x}\right\Vert ^{2}}{4};\begin{array}{c}
\left(\frac{d}{\alpha},\frac{2}{\alpha}\right)\\
\left(\frac{d}{\alpha}+\beta+1,\frac{2}{\alpha}\right),\left(\frac{d}{2},1\right)
\end{array}\right]\ ,\ \boldsymbol{x}\in\mathbb{R}^{d},\label{eq:fourier_kuttner_golubov}
\end{equation}
where $\Psi_{p,q}$ denotes the Fox-Wright generalized hypergeometric
function with $p$ numerator parameters and $q$ denominator parameters:
\begin{equation}
\Psi_{p,q}\left[u;\begin{array}{c}
(a_{1},b_{1}),\ldots,(a_{p},b_{p})\\
(c_{1},d_{1}),\ldots(c_{q},d_{q})
\end{array}\right]:=\sum^{\infty}_{n=0}\frac{\prod^{p}_{\ell=1}\Gamma(a_{\ell}+nb_{\ell})}{\prod^{q}_{\ell=1}\Gamma(c_{\ell}+nd_{\ell})}\frac{u^{n}}{n!}\ .\label{eq:fox_wright_function}
\end{equation}
\end{prop}

\begin{proof}
According to equation~\eqref{eq:kernel_fourier_transform}, the Fourier
transform of the kernel~\eqref{eq:kuttner_golubov_kernel} is given
by
\[
f(\boldsymbol{x})=\frac{1}{(2\pi)^{\frac{d}{2}}\left\Vert \boldsymbol{x}\right\Vert ^{\frac{d}{2}-1}}\int^{1}_{0}t^{\frac{d}{2}}(1-t^{\alpha})^{\beta}J_{\frac{d}{2}-1}(\left\Vert \boldsymbol{x}\right\Vert t)dt\ .
\]
Then, we can use the following series representation of the Bessel
function $J_{\nu}$ \citep[equation~(8) page~40]{watson1944bessel}\citep[10.2.2]{DLMF}:
\[
J_{\nu}(z)=\sum^{\infty}_{n=0}\frac{(-1)^{n}(z/2)^{\nu+2n}}{n!\Gamma(n+\nu+1)}
\]
to obtain
\begin{align*}
f(\boldsymbol{x}) & =\frac{1}{(2\pi)^{\frac{d}{2}}\left\Vert \boldsymbol{x}\right\Vert ^{\frac{d}{2}-1}}\int^{1}_{0}t^{\frac{d}{2}}(1-t^{\alpha})^{\beta}\sum^{\infty}_{n=0}\frac{(-1)^{n}(\left\Vert \boldsymbol{x}\right\Vert t/2)^{\frac{d}{2}-1+2n}}{n!\Gamma\!\left(n+\frac{d}{2}\right)}dt\\
 & =\frac{1}{2^{\frac{d}{2}-1}(2\pi)^{\frac{d}{2}}}\sum^{\infty}_{n=0}\frac{(-1)^{n}\left\Vert \boldsymbol{x}\right\Vert ^{2n}}{4^{n}n!\Gamma\!\left(n+\frac{d}{2}\right)}\int^{1}_{0}t^{2n+d-1}(1-t^{\alpha})^{\beta}dt\ .
\end{align*}
Then, using a change of variable,
\[
\int^{1}_{0}t^{2n+d-1}(1-t^{\alpha})^{\beta}dt=\frac{1}{\alpha}\int^{1}_{0}t^{\frac{2n+d}{\alpha}-1}(1-t)^{\beta}dt=\frac{1}{\alpha}\frac{\Gamma\!\left(\frac{2n+d}{\alpha}\right)\Gamma(\beta+1)}{\Gamma\!\left(\frac{2n+d}{\alpha}+\beta+1\right)}\ ,
\]
which gives
\begin{align*}
f(\boldsymbol{x}) & =\frac{1}{\alpha2^{\frac{d}{2}-1}(2\pi)^{\frac{d}{2}}}\sum^{\infty}_{n=0}\frac{(-1)^{n}\left\Vert \boldsymbol{x}\right\Vert ^{2n}}{4^{n}n!\Gamma\!\left(n+\frac{d}{2}\right)}\frac{\Gamma\!\left(\frac{2n+d}{\alpha}\right)\Gamma(\beta+1)}{\Gamma\!\left(\frac{2n+d}{\alpha}+\beta+1\right)}\\
 & =\frac{\Gamma(\beta+1)}{\alpha2^{\frac{d}{2}-1}(2\pi)^{\frac{d}{2}}}\sum^{\infty}_{n=0}\frac{\left(-\left\Vert \boldsymbol{x}\right\Vert ^{2}/4\right)^{n}}{n!}\frac{\Gamma\!\left(\frac{2n+d}{\alpha}\right)}{\Gamma\!\left(n+\frac{d}{2}\right)\Gamma\!\left(\frac{2n+d}{\alpha}+\beta+1\right)}\\
 & =\frac{\Gamma(\beta+1)}{\alpha2^{\frac{d}{2}-1}(2\pi)^{\frac{d}{2}}}\Psi_{1,2}\left[-\frac{\left\Vert \boldsymbol{x}\right\Vert ^{2}}{4};\begin{array}{c}
\left(\frac{d}{\alpha},\frac{2}{\alpha}\right)\\
\left(\frac{d}{\alpha}+\beta+1,\frac{2}{\alpha}\right),\left(\frac{d}{2},1\right)
\end{array}\right]\ .
\end{align*}
\end{proof}

\subsection{Symmetric beta kernels\label{subsec:symmetric_beta_kernels}}

Setting $\alpha=2$ in equation~\eqref{eq:kuttner_golubov_kernel}
gives the class of multivariate symmetric beta kernels \citep{marron1988canonical,duong2015beta}:
\begin{equation}
K(\boldsymbol{u})=k(\left\Vert \boldsymbol{u}\right\Vert )=(1-\left\Vert \boldsymbol{u}\right\Vert ^{2})^{\beta}\mathbbm{1}_{\{\left\Vert \boldsymbol{u}\right\Vert \leq1\}}\ ,\ \boldsymbol{u}\in\mathbb{R}^{d},\label{eq:symmetric_beta_kernel}
\end{equation}
which contains the multivariate parabolic, biweight, and triweight
kernels (see Table~\ref{tab:kuttner_golubov}).
\begin{cor}
\label{cor:fourier_symmetric_beta}The Fourier transform of the kernel
$K(\boldsymbol{u})=(1-\left\Vert \boldsymbol{u}\right\Vert ^{2})^{\beta}\mathbbm{1}_{\{\left\Vert \boldsymbol{u}\right\Vert \leq1\}}$,
$\boldsymbol{u}\in\mathbb{R}^{d}$, is given by
\begin{equation}
f(\boldsymbol{x})=\frac{\Gamma(\beta+1)}{2^{d}\pi^{\frac{d}{2}}\Gamma\!\left(\frac{d}{2}+\beta+1\right)}{}_{\,0}F_{1}\!\left(;\frac{d}{2}+\beta+1;-\frac{\left\Vert \boldsymbol{x}\right\Vert ^{2}}{4}\right)\ ,\ \boldsymbol{x}\in\mathbb{R}^{d},\label{eq:fourier_symmetric_beta}
\end{equation}
where $_{0}F_{1}$ denotes the confluent hypergeometric function with
$0$ numerator parameter and $1$ denominator parameter:
\begin{equation}
_{\,0}F_{1}(;b;u):=\sum^{\infty}_{n=0}\frac{1}{(b)_{n}}\frac{u^{n}}{n!}\ ,\label{eq:0F1}
\end{equation}
where the rising factorial $(b)_{n}$ is defined by $(b)_{n}=\Gamma(b+n)/\Gamma(b)$.
\end{cor}

\begin{proof}
Setting $\alpha=2$ in equation~\eqref{eq:fourier_kuttner_golubov}
gives
\begin{align*}
f(\boldsymbol{x}) & =\frac{\Gamma(\beta+1)}{2^{d}\pi^{\frac{d}{2}}}\sum^{\infty}_{n=0}\frac{\left(-\left\Vert \boldsymbol{x}\right\Vert ^{2}/4\right)^{n}}{n!}\frac{1}{\Gamma\!\left(n+\frac{d}{2}+\beta+1\right)}\\
 & =\frac{\Gamma(\beta+1)}{2^{d}\pi^{\frac{d}{2}}\Gamma\!\left(\frac{d}{2}+\beta+1\right)}\sum^{\infty}_{n=0}\frac{\left(-\left\Vert \boldsymbol{x}\right\Vert ^{2}/4\right)^{n}}{n!}\frac{1}{\left(\frac{d}{2}+\beta+1\right)_{n}}\\
 & =\frac{\Gamma(\beta+1)}{2^{d}\pi^{\frac{d}{2}}\Gamma\!\left(\frac{d}{2}+\beta+1\right)}{}_{\,0}F_{1}\!\left(;\frac{d}{2}+\beta+1;-\frac{\left\Vert \boldsymbol{x}\right\Vert ^{2}}{4}\right)\ .
\end{align*}
\end{proof}

\subsection{Askey kernels\label{subsec:askey_kernels}}

Setting $\alpha=1$ in equation~\eqref{eq:kuttner_golubov_kernel}
gives the Askey functions \citep{askey1973radial}:
\begin{equation}
K(\boldsymbol{u})=k(\left\Vert \boldsymbol{u}\right\Vert )=(1-\left\Vert \boldsymbol{u}\right\Vert )^{\beta}\mathbbm{1}_{\{\left\Vert \boldsymbol{u}\right\Vert \leq1\}}\ ,\ \boldsymbol{u}\in\mathbb{R}^{d},\label{eq:triangular_kernel}
\end{equation}
which include the multivariate triangular kernel as a particular case
(see Table~\ref{tab:kuttner_golubov}).
\begin{cor}
\label{cor:fourier_askey}The Fourier transform of the kernel $K(\boldsymbol{u})=(1-\left\Vert \boldsymbol{u}\right\Vert )^{\beta}\mathbbm{1}_{\{\left\Vert \boldsymbol{u}\right\Vert \leq1\}}$,
$\boldsymbol{u}\in\mathbb{R}^{d}$, is given by
\begin{equation}
f(\boldsymbol{x})=\frac{1}{2^{\frac{d}{2}-1}(2\pi)^{\frac{d}{2}}}\frac{\Gamma(\beta+1)\Gamma(d)}{\Gamma\!\left(\frac{d}{2}\right)\Gamma(d+\beta+1)}{}_{\,1}F_{2}\!\left(\frac{d+1}{2};\frac{d+1}{2}+\frac{\beta}{2},\frac{d}{2}+1+\frac{\beta}{2};-\frac{\left\Vert \boldsymbol{x}\right\Vert ^{2}}{4}\right),\ \boldsymbol{x}\in\mathbb{R}^{d},\label{eq:fourier_askey}
\end{equation}
where $_{1}F_{2}$ denotes the confluent hypergeometric function with
$1$ numerator parameter and $2$ denominator parameters:
\begin{equation}
_{\,1}F_{2}(a_{1};b_{1},b_{2};u):=\sum^{\infty}_{n=0}\frac{(a_{1})_{n}}{(b_{1})_{n}(b_{2})_{n}}\frac{u^{n}}{n!}\ .\label{eq:1F2}
\end{equation}
\end{cor}

\begin{proof}
Setting $\alpha=1$ in equation~\eqref{eq:fourier_kuttner_golubov}
and using the Legendre duplication formula $\Gamma(2z)=\frac{2^{2z-1/2}}{\sqrt{2\pi}}\Gamma(z)\Gamma(z+1/2)$
\citep[5.5.5]{DLMF} gives
\begin{align*}
f(\boldsymbol{x}) & =\frac{\Gamma(\beta+1)}{2^{\frac{d}{2}-1}(2\pi)^{\frac{d}{2}}}\sum^{\infty}_{n=0}\frac{\left(-\left\Vert \boldsymbol{x}\right\Vert ^{2}/4\right)^{n}}{n!}\frac{\Gamma(2n+d)}{\Gamma\!\left(n+\frac{d}{2}\right)\Gamma(2n+d+\beta+1)}\\
 & =\frac{\Gamma(\beta+1)}{2^{\frac{d}{2}+\beta}(2\pi)^{\frac{d}{2}}}\sum^{\infty}_{n=0}\frac{\left(-\left\Vert \boldsymbol{x}\right\Vert ^{2}/4\right)^{n}}{n!}\frac{\Gamma\!\left(n+\frac{d+1}{2}\right)}{\Gamma\!\left(n+\frac{d+1}{2}+\frac{\beta}{2}\right)\Gamma\!\left(n+\frac{d}{2}+1+\frac{\beta}{2}\right)}\\
 & =\frac{1}{2^{\frac{d}{2}-1}(2\pi)^{\frac{d}{2}}}\frac{\Gamma(\beta+1)\Gamma(d)}{\Gamma\!\left(\frac{d}{2}\right)\Gamma(d+\beta+1)}{}_{\,1}F_{2}\!\left(\frac{d+1}{2};\frac{d+1}{2}+\frac{\beta}{2},\frac{d}{2}+1+\frac{\beta}{2};-\frac{\left\Vert \boldsymbol{x}\right\Vert ^{2}}{4}\right)
\end{align*}
after simplifying the multiplicative constant using the Legendre duplication
formula again.
\end{proof}

\begin{rem}
Using the Gauss duplication formula $\Gamma(nz)=\frac{n^{nz-1/2}}{(2\pi)^{(n-1)/2}}\prod^{n-1}_{k=0}\Gamma\!\left(z+\frac{k}{n}\right)$,
$n\in\mathbb{N}^{*}$ \citep[5.5.6]{DLMF} makes it possible to express
the Fourier transform \eqref{eq:fourier_kuttner_golubov} in terms
of the $_{p}F_{q}$ hypergeometric function whenever $\alpha=2/n$
where $n\in\mathbb{N}^{*}$. More precisely,
\begin{align}
 & f(\boldsymbol{x})=\frac{1}{\alpha2^{\frac{d}{2}-1}(2\pi)^{\frac{d}{2}}}\frac{\Gamma(\beta+1)}{\Gamma\!\left(\frac{d}{2}\right)}\frac{\Gamma\!\left(\frac{d}{\alpha}\right)}{\Gamma\!\left(\frac{d}{\alpha}+\beta+1\right)}\sum^{\infty}_{m=0}\frac{\left(-\left\Vert \boldsymbol{x}\right\Vert ^{2}/4\right)^{m}}{m!}\frac{\prod^{n-1}_{k=1}\left(\frac{d}{2}+\frac{k}{n}\right)_{m}}{\prod^{n-1}_{k=0}\left(\frac{d}{2}+\frac{\beta+1}{n}+\frac{k}{n}\right)_{m}}\nonumber \\
 & =\frac{\Gamma(\beta+1)\Gamma\!\left(\frac{d}{\alpha}\right)}{\alpha2^{\frac{d}{2}-1}(2\pi)^{\frac{d}{2}}\Gamma\!\left(\frac{d}{2}\right)\Gamma\!\left(\frac{d}{\alpha}+\beta+1\right)}{}_{\,n-1}F_{n}{\scriptstyle \left(\frac{d}{2}\!+\!\frac{1}{n},\frac{d}{2}\!+\!\frac{2}{n},\ldots,\frac{d}{2}\!+\!\frac{n-1}{n};\frac{d}{2}\!+\!\frac{\beta+1}{n},\frac{d}{2}\!+\!\frac{\beta+2}{n},\ldots,\frac{d}{2}\!+\!\frac{\beta+n}{n};-\frac{\left\Vert \boldsymbol{x}\right\Vert ^{2}}{4}\right)}\label{eq:fourier_askey_2_n}
\end{align}
when $\alpha=2/n$ with $n\in\mathbb{N}^{*}$. Equation~\eqref{eq:fourier_askey_2_n}
generalizes both equation~\eqref{eq:fourier_symmetric_beta} ($\alpha=2$,
$n=1$) and equation~\eqref{eq:fourier_askey} ($\alpha=1$, $n=2$).
\end{rem}

\end{document}